\documentclass[11pt]{article}

\usepackage{fullpage}

\usepackage[T1]{fontenc}
\usepackage{amsmath, amsthm, amssymb, microtype, enumerate, stackrel, mathrsfs, tikz}
\usepackage[colorlinks, citecolor=blue, linkcolor=blue, bookmarks=true]{hyperref}
\usepackage[nameinlink, noabbrev, capitalize]{cleveref}
\usepackage{algpseudocode}

\usepackage{bbm}

\newtheorem{lemma}{Lemma}[section]
\newtheorem{theorem}[lemma]{Theorem}

\newtheorem{definition}[lemma]{Definition}
\newtheorem{claim}[lemma]{Claim}

\newcommand{\norm}[1]{\left\lVert#1\right\rVert}

\newcommand{\R}{\mathbb{R}}

\newcommand{\N}{\mathbb{N}}

\renewcommand{\emptyset}{\varnothing}
\renewcommand{\epsilon}{\varepsilon}

\renewcommand{\tilde}{\widetilde}

\renewcommand{\hat}{\widehat}

\DeclareMathOperator{\poly}{poly}
\DeclareMathOperator*{\E}{\mathbb{E}}
\DeclareMathOperator*{\Var}{Var}

\newcommand{\set}[1]{\left\{{#1}\right\}}
\newcommand{\B}{\set{0,1}}

\newcommand{\eps}{\epsilon}

\newcommand{\Tr}{\mathrm{Tr}}
\newcommand{\dmax}{d_{\mathrm{max}}}
\newcommand{\wmax}{w_{\mathrm{max}}}
\newcommand{\wmin}{w_{\mathrm{min}}}

\newcommand{\Ker}{\operatorname{ker}}
\newcommand{\Ber}{\operatorname{Bernoulli}}
\newcommand{\Image}{\operatorname{Im}}
\newcommand{\Span}{\operatorname{span}}
\newcommand{\onesvec}{\mathbf{1}}
\newcommand{\zerosvec}{\mathbf{0}}
\newcommand{\Sparsify}{\mathsf{Sparsify}}
\newcommand{\lm}{\rho}

\newcommand{\Dnote}[1]{{\color{green} \textbf{Dean note:} #1}}
\newcommand{\Jnote}[1]{{\color{red} \textbf{Jack note:} #1}}

\def\textprob#1{\textmd{\textsc{#1}}}

\usepackage{mathtools}

\newcounter{myalgctr}
\newenvironment{myalg}{
   \smallskip\noindent
   \refstepcounter{myalgctr}
   }{\par\smallskip}  
\numberwithin{myalgctr}{section}

\newtheoremstyle{named}{}{}{\itshape}{}{\bfseries}{.}{.5em}{\thmnote{#3}}
\theoremstyle{named}
\newtheorem{namedtheorem}{Theorem}[section]

\makeatletter
\DeclareRobustCommand\widecheck[1]{{\mathpalette\@widecheck{#1}}}
\def\@widecheck#1#2{%
	\setbox\z@\hbox{\m@th$#1#2$}%
	\setbox\tw@\hbox{\m@th$#1%
		\widehat{%
			\vrule\@width\z@\@height\ht\z@
			\vrule\@height\z@\@width\wd\z@}$}%
	\dp\tw@-\ht\z@
	\@tempdima\ht\z@ \advance\@tempdima2\ht\tw@ \divide\@tempdima\thr@@
	\setbox\tw@\hbox{%
		\raise\@tempdima\hbox{\scalebox{1}[-1]{\lower\@tempdima\box
				\tw@}}}%
	{\ooalign{\box\tw@ \cr \box\z@}}}
\makeatother

\makeatletter
\crefdefaultlabelformat{\leavevmode\sw@slant#2{\upshape#1}#3}
\makeatother

\begin{document}

\title{Spectral Sparsification via Bounded-Independence Sampling}

\author{Dean Doron\thanks{Research supported by a Motwani Postdoctoral Fellowship.} \\ ~~~~~Department of Computer Science ~~~~~\\ Stanford University \\ \texttt{ddoron@stanford.edu}  \and Jack Murtagh\thanks{Research supported by NSF grant CCF-1763299.} \\ School of Engineering \& Applied Sciences \\ Harvard University \\ \texttt{jmurtagh@g.harvard.edu} \\ \;  \and Salil Vadhan\thanks{Research supported by NSF grant CCF-1763299 and a Simons Investigator Award.} \\ School of Engineering \& Applied Sciences \\ Harvard University \\ \texttt{salil\textunderscore vadhan@harvard.edu}  \and David Zuckerman\thanks{Research supported in part by NSF Grant CCF-1705028 and a Simons Investigator Award (\#409864).} \\ ~~~~~Department of Computer Science~~~~~ \\ University of Texas at Austin \\ \texttt{diz@cs.utexas.edu}}

\date{}

\maketitle

\begin{abstract}
We give a deterministic, nearly logarithmic-space algorithm for mild spectral sparsification
of undirected graphs. Given a weighted, undirected graph $G$ on $n$ vertices described by a binary string of length $N$, an integer $k\leq \log n$, and an error parameter
$\eps > 0$, our algorithm runs in space $\tilde{O}(k\log (N\cdot\wmax/\wmin))$ where $\wmax$ and $\wmin$ are the maximum and minimum edge weights in $G$, and produces a weighted graph $H$ with $\tilde{O}(n^{1+2/k}/\eps^2)$ edges that
spectrally approximates $G$, in the sense of Spielmen and Teng \cite{ST04}, up to an error of $\eps$.

Our algorithm is based on a new bounded-independence analysis of Spielman and Srivastava's effective resistance based edge sampling algorithm \cite{SS08} and 
uses results from recent work on space-bounded Laplacian solvers \cite{MRSV17}. 
In particular, we demonstrate an inherent tradeoff (via upper and lower bounds) between the amount of (bounded) independence used in the edge sampling algorithm, denoted by $k$ above, and the resulting sparsity that can be achieved. 
 
\end{abstract}

\thispagestyle{empty}
\newpage


\section{Introduction} \label{sec:intro}
The graph sparsification problem is the following: given a weighted, undirected graph $G$, compute a graph $H$ that has very few edges but is a close approximation to $G$ for some definition of approximation. In general, graph sparsifiers are useful for developing more efficient graph-theoretic approximation algorithms. Algorithms whose complexity depend on the number of edges in the graph will be more efficient when run on the sparser graph $H$, and if $H$ approximates $G$ in an appropriate way, the result on $H$ may be a good approximation to the desired result on $G$. In this work, we present an algorithm that can be implemented deterministically in small space and achieves sparsification in the \emph{spectral sense} of Spielman and Teng \cite{ST04}. (See \cref{sec:intro_main_result} below for a more formal statement of our main result.)

\subsection{Background}
Motivated by network design and motion planning, Chew \cite{chew1989there} studied \emph{graph spanners}, which are sparse versions of graphs that approximately preserve the shortest distance between each pair of vertices. Bencz{\'u}r and Karger \cite{benczur1996approximating} defined \emph{cut sparsifiers} whose notion of approximation is that every cut of $H$ has size within a $(1\pm\eps)$ factor of the size of the corresponding cut in $G$.  They showed that every graph $G$ on $n$ vertices has a cut sparsifier $H$ with $O(n\cdot \log n/\epsilon^2)$ edges and gave a randomized algorithm for computing such cut sparsifiers. Their algorithm runs in nearly linear time (i.e., $\tilde{O}(m)$ where $m$ is the number of edges in $G$ and the $\tilde{O}(\cdot)$ notation hides polylogarithmic factors) and they used it to give a faster algorithm for approximating minimum $s\mbox{-}t$ cuts. 

Spielman and Teng introduced \emph{spectral sparsifiers}, which define approximation between the graph and its sparsifier in terms of the quadratic forms of their \emph{Laplacians} \cite{ST04}. The Laplacian of an  undirected graph is the matrix $L=D-A$ where $A$ is the adjacency matrix of the graph and $D$ is the diagonal matrix of vertex degrees (i.e. $D_{ii}$ equals the weighted degree of vertex $i$). $H$ is said to be an \emph{$\eps$-spectral approximation} of $G$ if for all vectors $v\in\mathbb{R}^n$, we have that $$v^{\top}\tilde{L}v\in (1\pm\eps)\cdot v^{\top}L v,$$ where $\tilde{L}$ and $L$ are the Laplacians of $H$ and $G$, respectively. Spectral sparsifiers generalize cut sparsifiers, which can be seen by observing that when $v\in\{0,1\}^n$, $v$ is the characteristic vector of some set of vertices $S\subseteq [n]$ and $v^{\top}L v$ equals the sum of the weights of the edges cut by $S$.

Spielman and Teng showed that all graphs have spectral sparsifiers with $O(n\cdot \log^{O(1)} n/\epsilon^2)$ edges and gave a nearly linear time randomized algorithm for computing them with high constant probability. Their spectral sparsifiers were a key ingredient that they used to develop the first nearly linear time algorithm for solving Laplacian systems. These fast Laplacian solvers spawned a flurry of improvements and simplifications \cite{cohen2014solving, kelner2013simple, koutis2014approaching, koutis2011nearly, kyng2016sparsified, kyng2016approximate, lee2013efficient, peng2014efficient} as well as extensions to directed graphs \cite{cohen2016faster, cohen2017almost, cohen2018solving} and to the space-bounded setting \cite{doron2017probabilistic, MRSV17, ahmadinejad2019high}. Spectral sparsification and the nearly linear time Laplacian solvers that use them have been critical primitives that have enabled the development of faster algorithms for a wide variety of problems including max flow \cite{liu2019faster, christiano2011electrical, cohen2017negative, kelner2014almost, lee2013new}, random generation of spanning trees, \cite{kelner2009faster, madry2014fast, schild2018almost}, and other problems in computer science \cite{orecchia2012approximating, koutis2009combinatorial}.

Spielman and Srivastava \cite{SS08} gave a spectral sparsification algorithm that both simplified and improved upon the algorithm of Spielman and Teng. They show that randomly sampling edges, independently with probabilities proportional to their \emph{effective resistances} produces a good spectral sparsifier with high probability. Viewing a graph as an electrical network, the effective resistance of an edge $(a,b)$ is the potential difference induced between them when a unit of current is injected at $a$ and extracted at $b$ (or vice versa). More formally, the effective resistance of an edge $(a,b)$ in a graph with Laplacian $L$ is
\begin{equation}
\label{eq:er_def}
R_{ab}=(e_a-e_b)^{\top}L^{+}(e_a-e_b)
\end{equation}
where $e_i$ denotes the $i$th standard basis vector and $L^{+}$ denotes the \emph{Moore-Penrose pseudoinverse} of $L$\footnote{$L^+$ is a matrix with the same kernel as $L$ that acts as an inverse of $L$ on the orthogonal complement of the kernel. See \cref{sec:moore_pen} for a formal definition.}. 

Spielman and Srivastava proved the following theorem.
\begin{theorem}[spectral sparsification via effective resistance sampling\footnote{In their original paper, \cite{SS08}, they fix the number of edges in the sparsifier in advance resulting in a slightly different theorem statement and analysis. The version we cite here and what we model our algorithm after was presented later in \cite{SNotes}.} \cite{SS08,SNotes}]
\label{thm:intro_ss_alg}
Let $G=(V,E,w)$ be a weighted graph on $n$ vertices and for each edge $(a,b)\in E$ with weight $w_{ab}$, define $p_{ab}=\min\{1,4\cdot\log n\cdot w_{ab}\cdot R_{ab}/\eps^2\}$, where $R_{ab}$ is the effective resistance of $(a,b)$ as defined in \cref{eq:er_def}.  Construct a sparsifier $H$ by sampling edges from $G$ independently such that each edge $(a,b)$ in $G$ is added to $H$ with probability $p_{ab}$. For edges that get added to $H$, reweight them with weight $w_{ab}/p_{ab}$. Let $L$ and $\tilde{L}$ be the Laplacians of $G$ and $H$, respectively. Then, with high probability,
\begin{enumerate}
\item $H$ has $O(n\cdot(\log n)/\eps^2)$ edges, and,
\item $\tilde{L}$ $\eps$-spectrally approximates $L$.
\end{enumerate}
Furthermore, this procedure can be implemented to run in time $\tilde{O}(\frac{m}{\eps^2}\cdot\log(w_{\mathrm{max}}/w_{\mathrm{min}}))$, where $m$ is the number of edges in $G$ and $w_{\mathrm{max}},w_{\mathrm{min}}$ are the maximum and minimum edge weights of $G$, respectively. 
\end{theorem}

The sparsity achieved by the Spielman and Srivastava sparsifiers was improved by Batson, Spielman and Srivastava \cite{batson2014twice}, who gave a deterministic algorithm for computing $\eps$-spectral sparsifiers with $O(n/\eps^2)$ edges, which is asymptotically optimal, however, their algorithm is less efficient, running in time $O(m\cdot n^3/\eps^2)$. Work on these optimal sparsifiers continued with another slightly faster deterministic algorithm \cite{zouzias2012matrix} followed by an $O(n^{2+\eps})$-time randomized algorithm \cite{allen2015spectral}, and culminating in the randomized algorithms of Lee and Sun who achieved almost-linear time \cite{lee2015constructing} and finally nearly-linear time \cite{lee2017sdp}.

\subsection{Our Main Result}
\label{sec:intro_main_result}
In this work we study the deterministic \emph{space complexity} of computing spectral sparsifiers. Our main result is a deterministic, nearly-logarithmic space algorithm for computing \emph{mild spectral sparsifiers}, that is, graphs with $O(n^{1+\alpha}/\eps^2)$ edges for any constant $\alpha>0$. 

\begin{theorem}[see also \cref{thm:derand}]
\label{thm:intro_main_result}
Let $G$ be a connected, weighted, undirected graph on $n$ vertices, $k \in \mathbb{N}$ an independence parameter 
and $\eps > 0$ an error parameter. There is a deterministic algorithm that on input $G$, $k$, and $\eps$,
outputs a weighted graph $H$ that is an
$\eps$-spectral sparsifier of $G$
and has $O(n^{1+2/k}\cdot(\log n)/\eps^2)$ edges.
The algorithm runs in space $O(k \log (N\cdot w)+\log (N\cdot w) \log\log (N\cdot w))$, where $w=\wmax/\wmin$ is the ratio of the maximum and minimum edge weights in $G$ and $N$ is the length of the input.
\end{theorem}

The closest analogue to spectral sparsifiers in the space-bounded derandomization literature is the \emph{derandomized square} of Rozenman and Vadhan \cite{RV05}, a graph operation that produces a sparse approximation to the square of a graph.\footnote{The \emph{square} of a graph $G$ is a graph on the same vertex set whose edges correspond to all walks of length 2 in $G$.} The derandomized square was introduced to give an alternative proof to Reingold's celebrated result that \textprob{Undirected S-T Connectivity} can be solved in deterministic logspace \cite{Reingold08}. Murtagh, Sidford, Reingold, and Vadhan \cite{MRSV17} showed that the derandomized square actually produces a \emph{spectral sparsifier} of the square of a graph and this was a key observation they used to develop a deterministic, nearly logarithmic space algorithm for solving Laplacian systems. Later the sparsification benefits of the derandomized square were also used in nearly logarithmic space algorithms for deterministically approximating random walk probabilities and for solving Laplacian systems in Eulerian directed graphs \cite{MRSV19, ahmadinejad2019high}.   

For a $d$-regular graph $G$ on $n$ vertices, its square $G^2$ has degree $d^2$ and the derandomized square computes an $\eps$-spectral approximation to $G^2$ with degree $O(d/\eps^2)$. On the other hand, applying our sparsification to $G^2$ results in an $\eps$-spectral approximation with on average $O(n^{\alpha}/\eps^2)$ edges adjacent to each vertex for any constant $\alpha$, which is independent of $d$ and much sparser when $d=\omega(n^{\alpha})$. Also, our algorithm can sparsify any undirected graph, not just squares. Our algorithm does not replace the derandomized square, however, because the derandomized square can be iterated very space efficiently, a property that is used in all of its applications thus far. Nevertheless, given the success of spectral sparsification and Laplacian solvers in the nearly-linear time context and the fruit borne of porting these techniques to the logspace setting, we are hopeful that our spectral sparsifiers will have further applications in derandomization of space-bounded computation.

\subsection{Techniques}

Our deterministic space-efficient algorithm is modeled after the effective resistance based sampling algorithm of Spielman and Srivastava (\cref{thm:intro_ss_alg}). Although the Spielman and Srivastava procedure is randomized and does not achieve optimal sparsity, the known algorithms that do (\cite{batson2014twice,zouzias2012matrix,allen2015spectral,lee2015constructing,lee2017sdp}) are more involved and often sequential in nature so do not seem as amenable to small-space implementations.  

To derandomize the Spielman-Srivastava algorithm, we follow the standard approach of first reducing the number of random bits used to logarithmic, and then enumerating over all random choices of the resulting algorithm. Following \cite{luby86,alon86},  a natural way to reduce the number of random bits used is to do the edge sampling only $k$-wise independently for some $k \ll |E|$ rather than sampling every edge independently from all other edges.

Let $k$ be our bounded-independence parameter. Namely, we are only guaranteed that
every subset of $k$ edges is chosen independently (with the right marginals), however there
may be correlations between the choices in tuples of size $k+1$. It is well known that
such a sampling can be performed using fewer random bits. By \cite{SS08}, we know
that $k = |E|$ will, with high probability, produce an $\eps$-spectral sparsifier
with $O(n\cdot\log n/\eps^2)$ edges in expectation. What about much
smaller values of $k$? In \cref{sec:upper-bound}, we prove the following:

\begin{theorem}[informal; see \cref{thm:main}]\label{thm1:informal}
Let $G$ be a connected weighted undirected graph on $n$ vertices with Laplacian $L$, $k \in \mathbb{N}$ an independence parameter 
and $\eps > 0$ an error parameter. Let
$H$ be the graph which is the output of Spielman and Srivastava's sampling-based
sparsification algorithm (\cref{thm:intro_ss_alg}), when the edge sampling is done in a $k$-wise
independent manner, and let $\tilde{L}$ be the Laplacian of $H$.
Then, with high constant probability, $\tilde{L}$ $\eps$-approximates $L$
and $H$ has $O(n^{1+2/k}\cdot(\log n)/\eps^2)$ edges.
\end{theorem}

A first thing to observe is that $k = \log n$ gives the same result
as in \cite{SS08}. More importantly, the above shows that the result \emph{interpolates}: Even
for a constant $k$, \cref{thm1:informal} gives a \emph{mild sparsification} that sparsifies
dense graphs to $O(n^{1+\alpha})$ expected edges, where $\alpha > 0$
is an arbitrarily small constant.

We prove \cref{thm1:informal} by extending the arguments in \cite{SS08, SNotes}.
For every edge $(a,b) \in E$, we define a random matrix $X_{ab}$ that
corresponds to the choice made by the sparsification algorithm, in such
a way that $X = \sum_{(a,b) \in E}X_{a,b}$ relates to the resulting 
Laplacian $\tilde{L}$.\footnote{Specifically, $X = L^{+/2}\tilde{L}L^{+/2}$, where $L^{+/2}$ is
the square-root of the pseudoinverse of $L$.} Let $\Pi$ be the orthogonal 
projection onto the image of $L$. 
Following \cite{SS08, SNotes}, we show that $\tilde{L}$ $\eps$-spectrally approximates $L$ (equivalently, that
$H$ is an $\eps$-spectral sparsifier for $G$) with high probability if $X-\Pi$ has bounded moments. Deriving a tail bound that relies on the
first $k$ moments alone, we can proceed with the analysis as if the $X_{ab}$'s 
were \emph{truly independent}. More specifically, we bound $\Tr(\E_{X}[(X-\Pi)^k])$ 
using a matrix concentration result due to Chen, Gittens and Tropp \cite{CGT12}.
For the complete details, as well as how our argument differs from \cite{SS08, SNotes},
see \cref{sec:upper-bound}. 

\paragraph{Getting a Deterministic Algorithm.}
\cref{thm1:informal} readily gives a simple, randomness-efficient 
algorithm, as $k$-wise independent sampling of edges only requires $O(k\cdot\log (N\cdot w))$ random bits \cite{J74,alon86} (See \cref{lem:sampling}). However, more work is needed to obtain a space-efficient deterministic algorithm. First, we need to be able to compute the marginal sampling probabilities, which depend on the effective resistances $R_{ab}$. Fortunately, the recent work of Murtagh et al.\ \cite{MRSV17} allows us to approximate the effective resistances using only $O(\log (N\cdot w) \log\log (N\cdot w))$ space and we show that the $k$-wise independent sampling procedure can tolerate the approximation. 

Next, to obtain a deterministic algorithm, we can enumerate over all possible random choices of the algorithm in space $O(k\cdot \log (N\cdot w))$ and compute a candidate sparsifier $H$ for each. We are guaranteed that at least one (indeed, most) of the resulting graphs $H$ is a good sparsifier for $G$ but how can we identify which one? To do this, it suffices for us, given Laplacians $L$ and $\tilde{L}$, to distinguish the case that $\tilde{L}$ is an $\epsilon$-spectral approximation of $L$ from the case that $\tilde{L}$ is not a $2\cdot\epsilon$-spectral approximation of $L$. We reduce that problem to that of approximating the spectral radius of 
$$
M = \left( \frac{(\tilde{L}-L)L^{+}}{\eps} \right)^{2}.
$$
where $L^+$ is the pseudoinverse of $L$, which can be approximated in nearly logarithmic space by \cite{MRSV17}. In fact, it will be sufficient to check whether the trace of a logarithmically high power of $M$
is below a certain threshold to deduce that the spectral radius of $M$ does not exceed $1$. 
In \cref{sec:verify}, we show that the latter case implies that $\tilde{L}$ indeed
$\eps$-approximates $L$. 

The deterministic, nearly logarithmic space Laplacian solver of \cite{MRSV17} only worked for \emph{multigraphs}, i.e. graphs with integer edge weights. To get our result for arbitrary weighted graphs, we extend the work of \cite{MRSV17} and give a deterministic, nearly logarithmic space Laplacian solver for arbitrary undirected weighted graphs. Combining this extension with the $k$-wise independent analysis of the edge sampling algorithm (\cref{thm1:informal}) and the verification procedure described above lets us prove our main result \cref{thm:intro_main_result}.

\subsection{Lower Bounds for Bounded-Independence Sampling}
Having established an upper bound on the amount of independence required
for the edge-sampling procedure (\cref{thm1:informal}), a natural goal would be to come up with 
a corresponding lower bound. \cref{thm1:informal} tells us that
in order to sparsify to $\tilde{O}(n^{1+\alpha})$ expected edges, we can use $k$-wise independent sampling for $k = 2/\alpha$. Can a substantially smaller choice of $k$ perform just as well? In \cref{sec:lower-bound}, we show that our upper bound of $k = 2/\alpha$ is tight up to a small constant factor. 

\begin{theorem}[informal; see \cref{thm:main-lower}]\label{thm3:informal}
For every small enough $\alpha > 0$ there exist infinitely many connected graphs $G = (V = [n],E)$ with all effective resistances equal
that are $d$-regular with $d=\Omega(n^\alpha)$ and a 
distribution $\mathcal{D} \sim \B^{|E|}$ that is $k$-wise independent for $k= \left\lfloor 4/3\alpha\right\rfloor$
with marginals $1/2$ that would fail to produce an $\epsilon$-spectral sparsifier of $G$ to within any $\epsilon > 0$ with high probability.
\end{theorem}

Our family of ``bad graphs'' will be dense graphs having large girth. Namely, 
given a girth $g$ and an integer $d \ge 3$, we consider graphs $G = (V=[n],E)$
satisfying $d \ge n^{\gamma/g}+1$ for some constant $0<\gamma < 2$ \cite{LUW95}. 
Getting an infinite family of graphs with $\gamma$ approaching $2$ (and
specifically attaining the \emph{Moore bound}), even non-explicitly, has been the subject of extensive study (see \cite{exoo2008dynamic} and references therein).
See also \cref{sec:dense} for a further discussion. Given a sparsification parameter $\alpha>0$, we set $k\approx\gamma/\alpha$ and take a graph $G$ on $n$ vertices with girth $g=k+1$ and degree $d>n^{\gamma/g}+1$.

Our construction of the distribution $\mathcal{D}$ is inspired by Alon and Nussboim \cite{AN08}: choose a partition of the vertices $V = V_0 \uplus V_1$ uniformly at random,
and for every edge $e = (u,v) \in E$, include it in the sample if and only if either $u,v \in V_0$ or $u,v \in V_1$.
Clearly, sampling edges according to $\mathcal{D}$ results in a disconnected graph almost surely. However, we show that $\mathcal{D}$ is indeed $k$-wise independent, relying on the fact
that the girth of $G$ is $k + 1$.

To obtain \cref{thm3:informal} we use the family of graphs given by Lazebnik et al.\ \cite{LUW95}
who obtained $\gamma = 4/3$. Indeed, any improvement in $\gamma$ would bring
our upper bound of $k \approx 2/\alpha$ and lower bound of $k \approx \gamma /\alpha$ closer together. 

\subsection{Open Problems}
An interesting open problem is to achieve improved sparsity, e.g. $O(n\cdot(\log n)/\eps^2)$ matching \cite{SS08}. Our algorithm would require space $\Omega(\log^2 n)$ to achieve this sparsity, due to setting $k=\Omega(\log n).$ We remark that previous work implies that this can be done in \emph{randomized} logarithmic space. Indeed, Doron et al. \cite{doron2017probabilistic} gave a randomized algorithm for solving Laplacian systems in logarithmic space (without $\log\log(\cdot)$ factors), and this implies that one can approximate effective resistances and hence implement the Spielman-Srivastava edge sampling with full independence in randomized logspace.  It is also an interesting question whether there is a nearly logspace algorithm (even randomized) that produces spectral sparsifiers of optimal sparsity (i.e., $O(n/\eps^2)$ edges). 

Finally, there has been recent progress on sparsifying Eulerian digraphs in the nearly-linear time literature \cite{cohen2016faster, cohen2017almost, cohen2018solving, chu2018graph}. Given the recent advance of a nearly-logarithmic space solver for Eulerian Laplacian systems \cite{ahmadinejad2019high}, an interesting question is sparsifying Eulerian graphs in small space.

\section{Preliminaries}

We will work with undirected weighted graphs, $G = (V,E,w)$, where $w$ is a vector of length $|E|$
and each edge $(a,b)\in E$ is associated with a positive weight $w_{ab}>0$.
At times we refer to undirected \emph{multigraphs}, which are weighted graphs where all of the weights are integers. 
The adjacency matrix of $G$ is a symmetric, real-valued matrix $A$ in which
$A_{ij}= w_{ij}$ if $(i,j) \in E$ and $A_{ij} = 0$ otherwise. 

For any matrix $A$, its
\emph{spectral norm} $\norm{A}$ is $\max_{\norm{x}=1}\norm{Ax}_2$, which is also
the largest singular value of $A$. 
For any square matrix $A$, its  \emph{spectral radius}, denoted $\lm(A)$, is the largest absolute value of its eigenvalues. 
When $A$ is real and symmetric,
the spectral norm equals the spectral radius. The spectral norm is sub-multiplicative,
i.e., $\norm{AB} \le \norm{A}\norm{B}$. We denote by $A^{\top}$ the transpose of $A$.
We denote by $\onesvec$ the all-ones
vector, by $\zerosvec$ the all-zeros vector, and 
$e_a$ is the vector with $1$ in the $a$-th coordinate
and $0$ elsewhere, where $e_a$'s dimension will be understood from context (i.e., $e_a$ is the $a$-th standard basis vector). 

The \emph{trace} of a matrix $A\in \mathbb{R}^{n\times n}$, is $\Tr(A)=\sum_{i\in[n]}A_{ii}$,
which also equals the sum of its eigenvalues. The trace is invariant under cyclic permutations, i.e., $\Tr(AB) = \Tr(BA)$. 
The expectation of a \emph{random matrix} is the matrix of the coordinate-wise expectations. More formally, if $A$ is a random matrix, then $\E[A]=\hat{A}$ where $\hat{A}_{ij}=\E[A_{ij}]$ for all $i,j\in[n]$. The trace and the expectation are both linear functions of a matrix and they commute. That is, for all random matrices $A$, we have $\Tr(\E[A])=\E[\Tr(A)]$ (see, e.g.,\cite{qiu2014cognitive}).

\subsection{PSD Matrices and Spectral Approximation}
A symmetric matrix $A\in\mathbb{R}^{n\times n}$ is \emph{positive semi-definite} (PSD), denoted $A \succeq 0$,  if
for every $x \in \mathbb{R}^{n}$ it holds that $x^{\top}Ax \ge 0$,
or equivalently, if all its eigenvalues are non-negative.  We write
$A \succeq B$ if $A-B \succeq 0$.

\begin{definition}\label{def:approx}
Let $A$ and $B$ be $n \times n$ symmetric PSD matrices. For 
a real $\eps > 0$, we say
that $A$ is an $\eps$-\emph{spectral approximation} of $B$, denoted $A \approx_{\eps} B$, if
$$
(1-\eps)B \preceq A \preceq (1+\eps)B.
$$
\end{definition}
When $A$ and $B$ share an eigenvector basis $v_1,\ldots,v_n$, \cref{def:approx} is
equivalent to requiring $(1-\eps)\mu_i \le \lambda_i \le (1+\eps)\mu_i$,
where $\lambda_{1},\ldots,\lambda_n$ are the
eigenvalues of $A$ corresponding to  $v_1,\ldots,v_n$ and $\mu_1,\ldots,\mu_n$ are
the eigenvalues of  $B$ corresponding to  $v_1,\ldots,v_n$.
\subsection{The Moore-Penrose Pseudoinverse}
\label{sec:moore_pen}

Let $A$ be any linear operator. The \emph{Moore-Penrose pseudoinverse} of
$A$, denoted $A^{+}$, is the unique matrix that satisfies the following:
\begin{enumerate}
\item $AA^{+}A = A$,
\item $A^{+}AA^{+} = A^{+}$, and,
\item both $AA^{+}$ and $A^{+}A$ are Hermitian.
\end{enumerate}
If $A = U \Sigma V^{\top}$ is the singular value decomposition (SVD)
of $A$, the pseudoinverse is given by $A^{+} = V\Sigma^{+}U^{\top}$
where $\Sigma^{+}$ is the matrix obtained by taking the reciprocal of each nonzero
diagonal element of $\Sigma$, and leaving the zeros intact. 
When $A$ is a symmetric PSD matrix, the SVD coincides with the eigen-decomposition
and so if $\lambda_{1},\ldots,\lambda_{n}$ are the eigenvalues of $A$
then $A^{+}$ shares the same eigenvector basis and has eigenvalues
$\lambda_{1}^{+},\ldots,\lambda_{n}^{+}$, where 
$$ \lambda_{i}^{+} = 
\begin{cases}
1/\lambda_i & \mathrm{if}~\lambda_i \neq 0, \\
0 & \mathrm{if}~\lambda_i = 0.
\end{cases}
$$ 
Also note that if $A$ is real then $A^{+}$
is real-valued as well.

A \emph{square root} of a matrix $A$ is any matrix $X$ that
satisfies $X^{2} = A$. When $A$ is symmetric and PSD, it has a unique symmetric PSD square root, which we write as $A^{1/2}$. If $A = U\Sigma U^{\top}$ is the eigen-decomposition of
$A$ then $A^{1/2} = U \sqrt{\Sigma} U^{\top}$ where
$\sqrt{\Sigma}$ is obtained by taking the square root
of each diagonal element of $\Sigma$. We denote by $A^{+/2}$ the
matrix
$( A^{+} )^{1/2} = ( A^{1/2} )^{+}$.

\subsection{The Graph Laplacian and Effective Resistance} \label{sec:intro-laplacian}

Given a graph $G$ on $n$ vertices with an adjacency matrix $A$ and degree matrix $D$ (i.e.,
$D$ is a diagonal matrix where $D_{ii} = \sum_{j=1}^{n}A_{ij}$ equals the weighted degree of vertex $i$ in $G$), the $\emph{Laplacian}$ of $G$ is the matrix
$$
L = D-A.
$$
For every undirected weighted graph $G = (V,E,w)$, its Laplacian $L$ is symmetric and PSD, with
smallest eigenvalue $0$. The zero eigenvalue has multiplicity one
if and only if $G$ is connected. In this case,
$\ker(L_G) = \Span(\set{\onesvec})$.
For every edge $(a,b) \in E$, define the \emph{edge Laplacian} of $(a,b)$ to be
$$
L_{ab} = (e_{a}-e_{b})(e_{a}-e_{b})^{\top} = (e_{b}-e_{a})(e_{b}-e_{a})^{\top}.
$$
Note that $L = \sum_{(a,b) \in E} w_{ab} \cdot L_{ab}$.

It is often helpful to associate $G$ with an electric circuit, where
an edge $(a,b) \in E$ corresponds to a resistor of resistance $1/w_{ab}$. 
For each pair of vertices $a$ and $b$, the \emph{effective resistance}
between them, denoted by $R_{ab}$, is the 
energy of the electrical flow that sends one unit of current from $a$ to $b$.
The effective resistance can be calculated using the pseudoinverse
of the Laplacian:
$$
R_{ab} = (e_{a}-e_{b})^{\top}L^{+}(e_{a}-e_{b}).
$$ 
(See \cite{Bol13} for more information on Laplacians and viewing graphs as electrical networks). A useful fact about effective resistances is Foster's Theorem:
\begin{theorem}[\cite{foster1949average}]\label{claim:foster}
For every undirected weighted graph $G = (V,E, w)$ on $n$ vertices it holds that
$$
\sum_{(a,b) \in E}w_{ab}\cdot R_{ab} = n-1.
$$
\end{theorem}

\subsection{Bounded-Independence Sampling} \label{sec:limited-independence}

Given a probability vector $p \in [0,1]^m$, let $\Ber(p)$ denote the distribution $X$
over $\B^m$ where the bits are independent and for each $i \in [m]$, $\E[X_i] = p_i$.
For a set $I  \subseteq [m]$ and a string $z \in \B^{m}$, we let $z|_{I} \in \B^{|I|}$
be the restriction of $z$ to the indices in $I$.

\begin{definition}
We say a distribution $X \sim \B^{m}$ is $k$-wise independent with marginals
$p \in [0,1]^m$ if for every set $I \subseteq [m]$ with $|I| \le k$, it holds that 
$X|_{I} = \Ber(p|_{I})$.
We refer to $X$ as a $k$-wise independent sample space with marginals
$p$.
\end{definition}

Consider $G = (V,E,w)$ with $|E| = m$.
Throughout, when we say \emph{sampling edges in a $k$-wise independent
manner}, we refer to the process of picking an element $x \in \B^{m}$ from a 
$k$-wise independent sample space uniformly at random and taking those edges $e \in E$ for which $x_e = 1$.

For $p \in [0,1]^m$ and a positive integer $t$, we define $\lfloor p \rfloor_t$
to be the vector $p'$ obtained by truncating every element of $p$ after
$t$ bits. Thus, for each $i \in [m]$, $p'_{i} = 2^{-t}\lfloor 2^{t}p_{i} \rfloor$, and so
$|p_i - p'_i| \le 2^{-t}$. 
The following lemma states that we can construct small $k$-wise independent 
sample spaces with any specified marginals.
\begin{lemma}[following \cite{J74,alon86}]
\label{lem:sampling}
For every $m,k,t \in \mathbb{N}$ and $p \in [0,1]^m$ there exists an explicit
$k$-wise independent distribution $X \sim \B^{m}$ with marginals $\lfloor p \rfloor_t$, that
can be sampled with $r=O(k \cdot \max\set{t,\log m})$ truly random bits. 
Furthermore, given $\rho\in\{0,1\}^r$, the element $x\in\mathrm{Supp}(X)$ corresponding to the random bits $\rho$ can be computed in $O(k \cdot \max\set{t,\log m})$ space. 
\end{lemma}
\subsection{Auxiliary Claims}

We will need the following claims, whose proofs we will defer to \cref{proofs-appendix}.
\begin{claim}
\label{lem:norm_of_sum}
Let $A,B, C$ be $n \times n$ symmetric PSD matrices and suppose that $B\preceq C$. Then
\[
\|A+B\|\leq \|A+C\|.
\]
\end{claim}

\begin{claim}
\label{lem:norm_is_er}
Let $G=(V,E,w)$ be an undirected weighted graph on $n$ vertices with Laplacian $L$. Fix $(a,b)\in E$ and recall that $L_{ab}=(e_a-e_b)(e_a-e_b)^{\top}$. Then,
\[
\left\|L^{+/2}L_{ab}L^{+/2}\right\| = R_{ab}.
\]
\end{claim}

\begin{claim}
\label{lem:proj_matrix}
Let $G = (V,E,w)$ be an undirected connected weighted graph on $n$ vertices with Laplacian $L$. Let $J$ be the $n\times n$ matrix with $1/n$ in every entry and define $\Pi=I-J$ (i.e. $\Pi$ is the projection onto $\Span(\onesvec)^{\perp}=\Image(L)$). Then, we have that
\[
\Pi=LL^+=L^+L=L^{+/2}LL^{+/2}.
\]
\end{claim}

\begin{claim}
\label{lem:mult_both_sides}
Let $A,B, C$ be symmetric $n\times n$ matrices and suppose $A$ and $B$ are PSD. Then the following hold 
\begin{enumerate}
\item $A\approx_{\eps} B\implies C^{\top}AC\approx_{\eps}C^{\top}BC$
\item If $\ker(C)\subseteq\ker(A)=\ker(B)$ then $A\approx_{\eps} B\iff C^{\top}AC\approx_{\eps}C^{\top}BC$
\end{enumerate}
\end{claim}

The proof of the following claim can be found in \cite{ST04}.
\begin{claim}[\cite{ST04}]
\label{claim:norm_of_pinv_stversion}
Let $G$ be an undirected, weighted graph on $n$ vertices with Laplacian $L$ and minimum weight $w_{\mathrm{min}}$. Then, the smallest nonzero eigenvalue of $L$ is at least $\min\set{\frac{8w_{\mathrm{min}}}{n^2},\frac{w_{\mathrm{min}}}{n}}$.
\end{claim}

\section{Sparsification via Bounded-Independence Sampling}\label{sec:upper-bound}

In \cref{sec:intro}, we briefly introduced the Spielman-Srivastava sparsification algorithm \cite{SS08}
based on (truly) independent edge sampling, with probabilities proportional to the effective resistances of the edges. In this section, we explore the tradeoff between the amount of independence used in the edge sampling process and the resulting sparsity that can be achieved. 

In particular, we analyze the algorithm $\Sparsify$ (see Figure 1).
The algorithm gets as input an undirected, weighted, dense graph $G=(V,E,w)$ on $n$ vertices, approximate effective resistances $\tilde{R}_{ab}$ for each edge $(a,b)\in E$, a bounded independence
parameter $k \le \log n$, a desired approximation error $\eps > 0$, and a parameter $\delta > 0$ governing the success probability, and outputs a sparser graph
$H$ whose Laplacian $\eps$-spectral approximates the Laplacian of $G$ with probability at least $1-2\delta$. 

\begin{figure}[h]
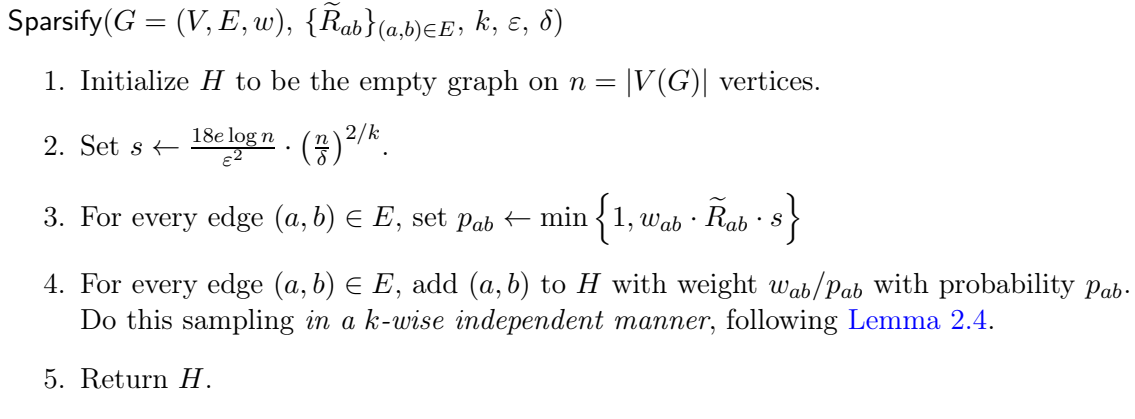
\label{alg:main}
\begin{myalg}
\centering
\fbox{
\parbox{0.9\textwidth}{
$\Sparsify$($G=(V,E,w)$, $\{\tilde{R}_{ab}\}_{(a,b)\in E}$, $k$, $\epsilon$, $\delta$)
\begin{enumerate}
    \item Initialize $H$ to be the empty graph on $n = |V(G)|$ vertices. 
    \item Set $s\leftarrow \frac{18 e\log n}{\eps^2} \cdot \left(\frac{n}{\delta}\right)^{2/k}$.
    \item For every edge $(a,b)\in E$, set $p_{ab}\leftarrow\min\set{1, w_{ab}\cdot \tilde{R}_{ab}\cdot s}$
    \item For every edge $(a,b)\in E$, add $(a,b)$ to $H$ with weight $w_{ab}/p_{ab}$ with probability $p_{ab}$. Do this sampling \emph{in a $k$-wise independent manner}, following \cref{lem:sampling}.
    \item Return $H$.
\end{enumerate}
}
}
\end{myalg}
\caption{Computing a spectral sparsifier via bounded independence sampling.}
\end{figure}

First we will analyze $\Sparsify$ for the case where the effective resistances are given exactly, i.e. $\tilde{R}_{ab}=R_{ab}$ for all $(a,b)\in E$. Then, in \cref{sect:approx_er} we will analyze the more general case where we are given approximations to the effective resistances. This latter case is useful algorithmically because more efficient algorithms are known for estimating effective resistances than for computing them exactly, both in the time-bounded and space-bounded settings \cite{SS08, MRSV17}.

\subsection{Sparsification With Exact Effective Resistances}
\label{sect:exact_er}
In this section we prove the following theorem about $\Sparsify$. 

\begin{theorem}[spectral sparsification via bounded independence]
\label{thm:main}
Let $G=(V,E, w)$ be an undirected connected weighted graph on $n$ vertices with Laplacian $L$ and effective resistances $R=\{R_{ab}\}_{(a,b)\in E}$. Let $0 < \eps < 1$, $0 < \delta < 1/2$ and let $k \le \log n$ be an even integer. Let $H$ be the output of $\Sparsify(G, R, k,\epsilon,\delta)$ and let $\tilde{L}$ be its Laplacian. Then, with probability at least $1-2\delta$ we have: 
\begin{enumerate}
\item $\tilde{L}\approx_{\epsilon}L$, and,
\item $H$ has $O\left(\frac{1}{\delta^{1+2/k}} \cdot \frac{\log n}{\eps^{2}}\cdot n^{1+\frac{2}{k}}\right)$ edges. 
\end{enumerate}
\end{theorem}
Spielman and Srivastava showed that by using truly independent sampling (i.e., $k = |E|$) in $\Sparsify$, one can compute an $\epsilon$-spectral sparsification of $G$ with $O(n\cdot\log n/\epsilon^2)$ edges, with high constant probability \cite{SS08}. One immediate consequence of \cref{thm:main} is that $\log n$-wise independent sampling suffices to match the sparsity that truly independent sampling achieves. Another consequence of \cref{thm:main} is that for any constant $0 < \alpha < 1$ and any constant $\gamma< \alpha/2$, for $k\approx2/(\alpha-2\gamma)$, $k$-wise independent sampling achieves a spectral sparsifier with error $\eps = n^{-\gamma}$  and $O(n^{1+\alpha})$ expected edges, with high constant probability.

 The proof of \cref{thm:main} is modeled after Spielman and Srivastava's argument \cite{SS08}. One difference is that the sparsification algorithm in \cite{SS08} fixes the number of edges to be sampled in advance rather than having the number of edges be a random variable. They then prove spectral approximation by reducing the problem to a question about concentration of random matrices, which they resolve with a matrix Chernoff bound due to Rudelson and Vershynin \cite{rudelson2007sampling}. We follow a variant of this argument for the case where the number of edges in the sparsifier is random and use a matrix concentration bound of Chen, Gittens, and Tropp \cite{CGT12}. This variant, for truly independent sampling, has appeared before in \cite{SNotes}. Our argument deviates in the proof of \cref{lem:trace_bound} to address the fact that we only use $k$-wise independent sampling.

We start by showing the sparsity guarantee in \cref{thm:main} indeed holds. Since the inclusion or exclusion of each edge in the sparsifier is a Bernoulli random variable, we can write the expected number of edges in it as 
\begin{align*} 
\sum_{(a,b)\in E}p_{ab}&\leq \sum_{(a,b)\in E}w_{ab}\cdot R_{ab}\cdot s\\ 
&=  (n-1)\cdot s\\ 
&=O\left(\frac{\log n}{\delta^{2/k}\eps^{2}}\cdot n^{\frac{2}{k}}\right)\cdot n, 
\end{align*}
where the second line follows from \cref{claim:foster}. By Markov's inequality, we can conclude: 
\begin{claim}\label{claim:item1}
Item (2) of \cref{thm:main} holds with probability at least $1-\delta$. 
\end{claim}
We prove item (1) of \cref{thm:main} by the following sequence of lemmas. Throughout, we let $G=(V,E,w)$, $L$, $\tilde{L}$, $\epsilon$ and $k$ as in \cref{thm:main} and
$s$, $R_{ab}$ and $p_{ab}$ as in $\Sparsify$. Let $\Pi = I-J$ be the orthogonal projection onto $\Image(L)$, as in \cref{lem:proj_matrix}. For each $(a,b)\in E$ we define the random matrix \[
X_{ab}=
\begin{cases}
\frac{w_{ab}}{p_{ab}}\cdot L^{+/2}L_{ab}L^{+/2} & \text{if we choose to include edge }(a,b)\text{ in Step 4 of }\Sparsify\\
0 & \mathrm{otherwise,}
\end{cases}
\]
So $\E[X_{ab}]=w_{ab}\cdot L^{+/2}L_{ab}L^{+/2}$ and the $X_{ab}$'s are $k$-wise independent. That is,
$$
\left( X_{e_1},\ldots,X_{e_k} \right) \equiv  X_{e1} \times \ldots \times X_{e_k} 
$$ 
for every $\set{e_1,\ldots,e_k} \subseteq E$.
Let $X=\sum_{(a,b)\in E}X_{ab}$. Also, recall that $L_{ab} = (e_a-e_b)(e_a-e_b)^{\top}$.

\begin{lemma}
\label{lem:reduction_to_X}
 Fix all the random choices in $\Sparsify$. Then, $\tilde{L}\approx_{\epsilon}L$ if and only if $X\approx_{\epsilon}\Pi$. 
\end{lemma}
\begin{proof}
By definition, $\tilde{L}\approx_{\epsilon}L$ if and only if
\[
(1-\epsilon) L\preceq \tilde{L}\preceq(1+\epsilon) L.
\]
Multiplying on both sides by $L^{+/2}$ and applying \cref{lem:proj_matrix} and \cref{lem:mult_both_sides}, we get that this is equivalent to 
\[
(1-\epsilon) \Pi\preceq L^{+/2}\tilde{L}L^{+/2}\preceq(1+\epsilon) \Pi.
\]
For each $(a,b)\in E$, let $Y_{ab}$ be the random variable 
\[
Y_{ab}=
\begin{cases}
\frac{w_{ab}}{p_{ab}}\cdot L_{ab} & \text{if we choose to include edge }(a,b)\text{ in Step 4 of }\Sparsify\\
0 & \mathrm{otherwise}
\end{cases}
\]
Notice that $\sum_{(a,b)\in E}Y_{ab}=\tilde{L}$. Thus, we have
\begin{align*}
    X&=\sum_{(a,b)\in E}X_{ab}\\
    &=L^{+/2}\left(\sum_{(a,b)\in E}Y_{ab}\right)L^{+/2}\\
    &=L^{+/2}\tilde{L}L^{+/2}.
\end{align*}
It follows that $\tilde{L}\approx_{\epsilon}L$ if and only if $X\approx_{\epsilon}\Pi$.
\end{proof}

\begin{lemma}
\label{lem:XapproxPi}
Fix all the random choices in $\Sparsify$ and assume $k$ is even. Then, $X\approx_{\epsilon}\Pi$ if $(X-\Pi)^k\preceq (\epsilon\cdot\Pi)^k$. 
\end{lemma}
\begin{proof}
First we observe that $X$ and $\Pi$ share a common eigenbasis. $\onesvec$ is in the kernel of both $\Pi=I-J$ and $X$. Let $v_2,\ldots, v_n$ be orthogonal eigenvectors of $X$ in $\Span(\set{\onesvec})^{\perp}$ with eigenvalues $\lambda_2,\ldots,\lambda_n$, respectively. These are all also eigenvectors of $\Pi$ since $\Span(\set{\onesvec})^{\perp}$ is an eigenspace of $\Pi$ of eigenvalue 1. Assume $(X-\Pi)^k\preceq (\epsilon\cdot\Pi)^k$. Since $v_i$ is an eigenvector of both $X$ and $\Pi$, we have $(\lambda_i-1)^k\leq\eps^k$. Since $k$ is even, it follows that $|\lambda_i-1|\leq\eps$, i.e. $1-\eps\leq\lambda_i\leq1+\eps$. Since this holds for all eigenvalues $\lambda_2,\ldots,\lambda_n$ of $X$ and the corresponding eigenvalues of $\Pi$ are $1$, we conclude that $X\approx_{\eps} \Pi$.

\end{proof}
\begin{lemma}
\label{lem:move_to_trace}
Fix all the random choices in $\Sparsify$ and assume $k$ is even. Then, $(X-\Pi)^k\preceq(\epsilon\cdot\Pi)^k$ if $\Tr\left((X-\Pi)^k\right)\leq\epsilon^k$.
\end{lemma}
\begin{proof}
Since $k$ is even, $(X-\Pi)^k$ has non-negative eigenvalues. If $\Tr\left((X-\Pi)^k\right)\leq\epsilon^k$ then the sum of the eigenvalues of $(X-\Pi)^k$ is at most $\epsilon^k$ and hence the largest eigenvalue of $(X-\Pi)^k$ is at most $\epsilon^k$. Since $\Ker(X)\supseteq\Ker(\Pi)$ and all nonzero eigenvalues of $(\epsilon\cdot\Pi)^k$ equal $\epsilon^k$, it follows that $(X-\Pi)^k\preceq(\epsilon\cdot\Pi)^k$.
\end{proof}
\begin{lemma}
\label{lem:prob_of_trace}
Assuming $k$ is even, it holds that
$\Pr_{X}[\Tr\left((X-\Pi)^k\right)>\epsilon^k]\leq \frac{1}{\eps^k}  \Tr\left(\E_{X}[(X-\Pi)^k]\right)$.
\end{lemma}
\begin{proof}
Since $\Tr((X-\Pi)^k)$ is nonnegative (due to $k$ being even), Markov's inequality gives 
\[
\Pr_{X}[\Tr\left((X-\Pi)^k\right)>\epsilon^k]~\leq~ \frac{1}{\eps^k} \cdot \E_{X}[\Tr((X-\Pi)^k)].
\]
Noting that the trace and the expectation commute completes the proof. 
\end{proof}
\begin{lemma}
\label{lem:trace_bound}
It holds that
$\Tr\left(\E_{X}[(X-\Pi)^k]\right) \leq n\cdot \left(\frac{18e\log n}{s}\right)^{k/2}$.
\end{lemma}
To prove \cref{lem:trace_bound}, we will use the following theorem of Chen, Gittens, and Tropp.

\begin{theorem}[\cite{CGT12}]
\label{thm:normbound}
Let $W_{1},\ldots,W_{m}$ be independent, random, symmetric $n \times n$ matrices. Fix $k\geq 2$ and let $r=\max\{k,2\log n\}$. Then,
\[
\left(\E\left[\left\|\sum_{i\in[m]}W_{i}\right\|^{k}\right]\right)^{1/k}\leq \sqrt{e\cdot r}\cdot\left\|\sum_{i\in[m]}\E[W_{i}^2]\right\|^{1/2}+2e\cdot r\cdot\left(\E\left[\max_{i\in[m]}\|W_{i}\|^{k}\right]\right)^{1/k}.
\]
\end{theorem}
\begin{proof}[Proof of \cref{lem:trace_bound}]
Define $Z_{ab}=X_{ab}-w_{ab}\cdot L^{+/2}L_{ab}L^{+/2}$ and let $Z=\sum_{(a,b)\in E}Z_{ab}=X-\Pi$. Our goal is to bound $\E[Z^k]$. To do this, we define $$\set{\hat{Z}_{ab}}_{(a,b) \in E}$$ to be identically distributed to the $Z_{ab}$'s, except that the $\hat{Z}_{ab}$ random variables are truly independent, instead of only $k$-wise. 
More specifically, if we let $\hat{X}_{ab}$ be defined the same way as $X_{ab}$ but this time we sample the edges in Step 4 of $\Sparsify$ truly independently (with marginals $p_{ab}$), then $\hat{Z}_{ab}=\hat{X}_{ab}-w_{ab}\cdot L^{+/2}L_{ab}L^{+/2}$.
Analogously, let $\hat{Z}=\sum_{(a,b)\in E}\hat{Z}_{ab}$.

The key point to notice is that
both $\E[\hat{Z}^k]$ and $\E[Z^k]$ can each be written as a sum of products
of at most $k$ random variables. As the $Z_{ab}$'s are $k$-wise independent, we have: 
\begin{claim}\label{claim:kwise}
$\E[\hat{Z}^k] = \E[Z^k]$.	
\end{claim}

Towards bounding $\Tr(\E[\hat{Z}^k])=\E[\Tr(\hat{Z}^k)]$, we first bound $\E[\|\hat{Z}^{k}\|]$. Then, we use the fact that for all symmetric $n\times n$ matrices $M$ we have $\Tr(M)\leq n\cdot \|M\|$ so 
\[
\Tr(\E[\hat{Z}^k])\leq n\cdot \E[\|\hat{Z}^k\|] \leq n\cdot \E[\|\hat{Z}\|^k],
\]
where the latter inequality is by the submultiplicity of the spectral norm. Since $\hat{Z}=\sum_{(a,b)\in E}\hat{Z}_{ab}$, we can bound the right-hand side by applying \cref{thm:normbound} to the $\hat{Z}_{ab}$'s. To bound the two terms on the right-hand side we make use of the following two claims. 

\begin{claim}
\label{claim:norm}
For every $(a,b) \in E$ and every matrix in the support of $\hat{Z}_{ab}$, it holds that
$\left\|\hat{Z}_{ab}\right\|\leq \frac{1}{s}$. 
\end{claim}
\begin{proof}
Observe that if $p_{ab}=1$ then $\hat{Z}_{ab}=0$. If $p_{ab}<1$,
\begin{align*}
\left\|\hat{Z}_{ab}\right\|&\leq \max\set{\left\|\left(\frac{1}{p_{ab}}-1\right)\cdot w_{ab}\cdot L^{+/2}L_{ab}L^{+/2}\right\|,\left\|-w_{ab}\cdot L^{+/2}L_{ab}L^{+/2}\right\|}\\
&\leq \frac{1}{p_{ab}} \cdot \left\| w_{ab}\cdot L^{+/2}L_{ab}L^{+/2}\right\| \\
&=\frac{w_{ab}\cdot R_{ab}}{p_{ab}} ~~~~~~~~~\text{(\cref{lem:norm_is_er})}\\
&=\frac{1}{s}.
\end{align*}
\end{proof}
\begin{claim}
\label{claim:variance}
For every $(a,b) \in E$
it holds that
$\E\left[\hat{Z}_{ab}^2\right]\preceq \frac{w_{ab}}{s}\cdot L^{+/2}L_{ab}L^{+/2}.$
\end{claim}
\begin{proof}
As $\E[\hat{Z}_{ab}] = 0$ we can write
\begin{align*}
    \E[\hat{Z}_{ab}^{2 }]&=\Var[X_{ab}]\\
    &=\frac{1}{p_{ab}^2}\cdot\Var\left[\mathrm{Ber}(p_{ab})\right]\cdot \left(w_{ab}\cdot L^{+/2}L_{ab}L^{+/2}\right)^2\\
    &=\left(\frac{1}{p_{ab}}-1\right)\cdot \left(w_{ab}\cdot L^{+/2}L_{ab}L^{+/2}\right)^2\\
    &=\left(\frac{1}{p_{ab}}-1\right)\cdot w_{ab}^2\cdot L^{+/2}(e_a-e_b)(e_a-e_b)^{\top}L^{+}(e_a-e_b)(e_a-e_b)^{\top}L^{+/2}\\
    &=\left(\frac{1}{p_{ab}}-1\right)\cdot w_{ab}^2\cdot L^{+/2}(e_a-e_b)R_{ab}(e_a-e_b)^{\top}L^{+/2}\\
    &=R_{ab}\cdot w_{ab}^2\cdot \left(\frac{1}{p_{ab}}-1\right)\cdot L^{+/2}L_{ab}L^{+/2}.
\end{align*}
Note that if $p_{ab}=1$ then the above expectation is 0. If $p_{ab}<1$ then 
\[
R_{ab}\cdot w_{ab}^2\cdot\left(\frac{1}{p_{ab}}-1\right)\cdot L^{+/2}L_{ab}L^{+/2}
\preceq \frac{w_{ab}}{s}\cdot L^{+/2}L_{ab}L^{+/2}.
\]
\end{proof}

Now we can bound the first term on the right-hand side of \cref{thm:normbound}. Together, \cref{claim:variance} and \cref{lem:norm_of_sum} give: 
\[
\left\|\sum_{(a,b)\in E}\E[\hat{Z}_{ab}^{2}]\right\|
\leq \left\|\frac{1}{s}\cdot\sum_{(a,b)\in E} w_{ab}\cdot L^{+/2}L_{ab}L^{+/2}\right\|.
\]
Now, recall that $\sum_{(a,b)\in E}w_{ab}\cdot L^{+/2}L_{ab}L^{+/2}=L^{+/2}LL^{+/2}=\Pi$ and $\|\Pi\|=1$, so
\[
\left\|\sum_{(a,b)\in E}\E[\hat{Z}_{ab}^{2}]\right\|^{1/2}\leq \frac{1}{\sqrt{s}}.
\]

To bound the second term of \cref{thm:normbound}, we apply \cref{claim:norm} to get
\[
\E\left[\max_{(a,b)\in E}\|\hat{Z}_{ab}\|^k\right]^{1/k}\leq \frac{1}{s}.
\]

Set $r=\max\set{k,2\log n} = 2\log n$. Combining the bounds on the two terms together and applying \cref{thm:normbound} gives
\begin{align*}
    \left(\E\left[\left\|\hat{Z}\right\|^{k}\right]\right)^{1/k}&\leq \sqrt{\frac{e\cdot r}{s}}+\frac{2e\cdot r}{s}\\
    &\leq 3\cdot \sqrt{\frac{2e\log n}{s}},
\end{align*}
when $s>e\cdot r$. Raising both sides to the $k$-th power and using the sub-multiplicativity of the spectral norm, we get
\[
   \E\left[\left\|\hat{Z}^{k}\right\|\right]\leq \E\left[\left\|\hat{Z}\right\|^{k}\right] \le  \left(\frac{18e\log n}{s}\right)^{k/2}. 
\]
For all symmetric $n\times n$ matrices $M$ we have $\Tr(M)\leq n\cdot \|M\|$ so by the monotonicity of expectation we get
\[
\Tr(\E[\hat{Z}^k])=\E\left[\Tr(\hat{Z}^{k})\right]\leq n\cdot \left(\frac{18e\log n}{s}\right)^{k/2}.
\]

By \cref{claim:kwise},
\[ 
\E\left[\Tr(Z^{k})\right]=\Tr(\E[Z^k])\leq n\cdot \left(\frac{18e\log n}{s}\right)^{k/2},
\]
and so
\[ \frac{1}{\eps^k} \cdot \Tr(\E[(X-\Pi)^k])\leq \frac{n}{\eps^k}\cdot \left(\frac{18e\log n}{s}\right)^{k/2}.
\]
\end{proof}

Now we can prove the main theorem of this section.
\begin{proof}[Proof of \cref{thm:main}]
From \cref{lem:reduction_to_X}, \cref{lem:XapproxPi}, \cref{lem:move_to_trace} and \cref{lem:prob_of_trace} we have that $\tilde{L}\approx_{\epsilon}L$ except with probability at most 
$$
\frac{1}{\eps^{k}} \cdot \Tr(\E[(X-\Pi)^k]).
$$ 
By \cref{lem:trace_bound} we have 
\[\frac{1}{\eps^k} \cdot \Tr(\E[(X-\Pi)^k])\leq \frac{n}{\eps^k}\cdot \left(\frac{18e\log n}{s}\right)^{k/2}.
\]
The above is upper bounded by $\delta$ whenever
\[
s\geq \frac{18e\log n \cdot n^{2/k}}{\delta^{2/k}\eps^2},
\]
which is how we set $s$ in $\Sparsify$. Combining this with
\cref{claim:item1}, the theorem follows by the union bound.
\end{proof}

\subsection{Sparsification With Approximate Effective Resistances}
\label{sect:approx_er}
Spielman and Srivastava showed that the original version of spectral sparsification through effective resistance sampling (with fully independent sampling and fixing the number of edges in advance) is robust to small changes in the sampling probabilities. In this section we show the same is true of $\Sparsify$. As said, this is useful because more efficient algorithms are known for estimating effective resistances than for computing them exactly, and we will also use this fact 
for our space-bounded algorithm for sparsification in \cref{sec:algorithm}. 

The lemma below says that if we only have small multiplicative approximations to the effective resistances then the guarantees of \cref{thm:main} still hold with a small loss in the sparsity.
\begin{lemma}
\label{lem:approx_er}
Let $G=(V,E,w)$ be an undirected connected weighted graph on $n$ vertices with Laplacian $L$. Let $0 < \eps < 1$,
$0 < \delta < 1/2$ and let $k \le \log n$ be an even integer. For each $(a,b)\in E$, let $\tilde{R}_{ab}$ be such that  
\[
(1-\alpha)\cdot R_{ab}\leq \tilde{R}_{ab}\leq (1+\alpha)\cdot R_{ab},
\]
where $R_{ab}$ is the effective resistance of $(a,b)$ and $0 < \alpha < 1$. Let $\tilde{R}=\{\tilde{R}_{ab}/(1-\alpha)\}_{(a,b)\in E}$. Let $H$ be the output of $\Sparsify (G,\tilde{R},k,\epsilon,\delta)$ and
let $\tilde{L}$ be its Laplacian. Then, with probability at least $1-2\delta$ we have: 
\begin{enumerate}
\item $\tilde{L}\approx_{\epsilon}L$, and,
\item $H$ has $O\left(\frac{1+\alpha}{1-\alpha} \cdot \frac{1}{\delta^{1+2/k}} \cdot \frac{\log n}{\eps^{2}}\cdot n^{1+\frac{2}{k}}\right)$ edges. 
\end{enumerate}
\end{lemma}
\begin{proof}
Using $\tilde{R}_{ab}/(1-\alpha)$ in $\Sparsify$, our sampling probabilities become
\[
\tilde{p}_{ab}=\min\{1,w_{ab}\cdot \tilde{R}_{ab}\cdot s/(1-\alpha)\}.
\]
This means that the expected sparsity of the resulting graph is
\begin{align*}
\sum_{(a,b)\in E}\tilde{p}_{ab}&\leq \sum_{(a,b)\in E}s\cdot w_{ab}\cdot \tilde{R}_{ab}/(1-\alpha)\\
&\leq s\cdot \frac{1+\alpha}{1-\alpha} \cdot \sum_{(a,b)\in E}w_{ab}\cdot R_{ab}\\
&=s\cdot\frac{1+\alpha}{1-\alpha}\cdot (n-1).
\end{align*}
Notice that by feeding $\Sparsify$ $\tilde{R}_{ab}/(1-\alpha)$ rather than $\tilde{R}_{ab}$, we guarantee that the approximate effective resistance is an upper bound on the true effective resistance and hence the approximate sampling probability is an upper bound on the true sampling probability. In particular, this implies that if $p_{ab}=1$ then $\tilde{p}_{ab}=1$. 

 Note that in \cref{lem:reduction_to_X} through \cref{lem:trace_bound}, the expectations of $X_{ab},Z_{ab},$ and $\hat{Z}_{ab}$ do not depend on the sampling probabilities. The sampling probabilities come up when we bound two terms of the concentration bound in \cref{thm:normbound}. However, because of our guarantee that $\tilde{p}_{ab}\geq p_{ab}$, we get the same results. The calculation we used for the first term, given in \cref{claim:variance}, now yields 
\begin{align*}
    \E[\hat{Z}_{ab}^{2}]&=R_{ab}\cdot w_{ab}^2 \cdot\left(\frac{1}{\tilde{p}_{ab}}-1\right)\cdot L^{+/2}L_{ab}L^{+/2}\\
    &\leq \frac{\tilde{R}_{ab}}{1-\alpha}\cdot w_{ab}^2 \cdot\left(\frac{1}{\tilde{p}_{ab}}-1\right)\cdot L^{+/2}L_{ab}L^{+/2}.
\end{align*}
Again, when $\tilde{p}_{ab}=1$, the above is 0, and otherwise we have $\tilde{p}_{ab} = w_{ab}\cdot \tilde{R}_{ab} \cdot s/(1-\alpha)$. Thus,
\[
\E[\hat{Z}_{ab}^{2}]\preceq w_{ab}\cdot L^{+/2}L_{ab}L^{+/2},
\]
which is exactly the bound in \cref{claim:variance}. Similarly, when adapting \cref{claim:norm} to the switch to $\tilde{p}_{ab}$, we incur no loss. We have
\begin{align*}
\|\hat{Z}_{ab}\|&\leq \frac{1}{\tilde{p}_{ab}}\cdot\|w_{ab}\cdot L^{+/2}L_{ab}L^{+/2}\|\\
&=\frac{w_{ab}\cdot R_{ab}}{\tilde{p}_{ab}}\\
&\leq \frac{w_{ab}\cdot R_{ab}}{p_{ab}}\\
&=\frac{1}{s},
\end{align*}
which matches the original bound. 

In fact, this lemma holds with a slightly weaker assumption. Notice that we used the fact that $(1-\alpha)\cdot R_{ab}\leq \tilde{R}_{ab}$ but for the upper bound on the approximate effective resistances, we only need the weaker inequality:
 \[
\sum_{(a,b)\in E}w_{ab}\cdot\tilde{R}_{ab}\leq (1+\alpha)\cdot\sum_{(a,b)\in E}w_{ab}\cdot R_{ab}
\]
for the argument above to go through. 

\end{proof}

Note that we could equivalently define $\Sparsify$ to take approximate sampling probabilities as input (i.e.,  $(1-\alpha) p_{ab}\leq \tilde{p}_{ab}\leq (1+\alpha) p_{ab}$)  rather than $\alpha$-approximate effective resistances and the same lemma applies.

\section{Lower Bounds for Bounded-Independence Sampling}\label{sec:lower-bound}

In this section we prove a lower bound for sampling-based bounded independence
sparsification. Our lower bound will hold even for unweighted, simple, regular graphs in which
all the effective resistances are the same, so for this section, assume $G = (V = [n],E)$ is such a graph. In \cref{sec:upper-bound} we measure sparsity in terms of the number of edges in the graph. We use this measure rather than average degree because in weighted graphs, the degree of a vertex $v$ typically refers to the sum of the weights of the edges incident to $v$, whereas in sparsification algorithms we are trying to minimize the number of edges incident to $v$, regardless of their weight. In this section, we will sometimes refer to average degree rather than number of edges. When we refer to the average degree of a weighted graph, we mean the average number of edges incident to each vertex. For simple, unweighted graphs, these quantities are the same. 

Fix some $\alpha > 0$.
\cref{thm:main} tells us that if we want to sparsify $G$
to within error $\eps$ and expected degree $s = O\left(\frac{\log n}{\eps^{2}}n^{\alpha}\right)$,
we can do so by sampling each edge with probability $p = \frac{s\cdot(n-1)}{|E|}$ in a
$k$-wise independent manner, where $k = \frac{2}{\alpha}$ (more precisely,
the next even integer).\footnote{We used the fact that for every $(a,b) \in E$, $p_{ab} \leftarrow
\min\set{1,R_{ab} s} = R_{ab} s = R\cdot s$, which can be argued as follows. 
When all effective resistances equal $R$, we have $R=(n-1)/|E|$ due to
\cref{claim:foster}. Now, if $G$ has $n \cdot s$ edges or fewer, then it already achieves the desired sparsity so without loss of generality we can assume that $|E|>n\cdot
s$. Hence, $R\cdot s<(n-1) s/n  s<1$. 
Also, the resulting
graph should indeed be a weighted one, however all its weights will be 
the same, $1/p$.
} We now prove that $k \ge 4/3\alpha$ is essential for
such a sampling procedure, at least for constant $\alpha$.

\begin{theorem}[lower bound for spectral sparsification via bounded independence]\label{thm:main-lower} 
Fix $c > 0$.
For every 
$\alpha \le 4/15$, there exist infinitely many $n$'s for
which the following holds.	

There exists a connected graph $G = (V = [n],E)$ whose effective resistances are all equal
and a distribution $\mathcal{D} \sim \B^{|E|}$ that is $k$-wise
independent for $k= \left\lfloor4/3\alpha\right\rfloor$
with marginals $1/2$ that
would fail to sparsify $G$ to within any error $\eps > 0$
and expected degree $s = c \log n \cdot  n^{\alpha_0}$,
where $\alpha_0 \ge (1-2\alpha)\alpha$. 

More specifically, sampling a subgraph of $G$ according to $\mathcal{D}$ would
result in a disconnected graph with probability at least $1-2^{1-n}$.
\end{theorem}


%

We note that a disconnected graph fails to be a good spectral sparsifier of a connected graph,
which is implicit in \cref{thm:main-lower}. Formally:
\begin{claim}
\label{lem:cc_spec_approx}
Let $G$ and $\tilde{G}$ be undirected graphs on $n$ vertices with Laplacians $L$ and $\tilde{L}$, respectively. If $G$ is connected and $\tilde{G}$ is disconnected then $\tilde{L}\not\approx_{\epsilon} L$ for any $\epsilon>0$.
\end{claim}
We give a proof of \cref{lem:cc_spec_approx} in \cref{proofs-appendix}.

%
%

\subsection{Moore-Like Graphs With a Given Girth}\label{sec:dense}

Toward proving \cref{thm:main-lower}, we will need, for every bounded-independence parameter $k$, an infinite
family of graphs satisfying certain properties.
Recall that the \emph{girth} of a graph $G$ is the length
of the shortest cycle in $G$. We will need an infinite family of girth-$g$ graphs having large degree. Formally:
\begin{definition}
Given $\gamma > 0$ and $g \colon \N \rightarrow \N$,
a family of graphs $\set{G_{i} = (V_{i}=[n_i],E_i)}_{i \in \N}$ is
$(g,\gamma)$\emph{-Moorish} if for every $i \in \N$, $G_{i}$ is connected, has girth
$g(n_i)$ and is $d$-regular for $d \ge n_{i}^{\gamma/g(n_i)}+1$.
\end{definition}
The problem of finding such families of graphs, or even
proving their existence in some regime of parameters, has been
widely studied in extremal graph theory. A simple counting argument (\cite{ES63}, see also \cite{Bol13}) shows
that $(g,\gamma)$-Moorish families of graphs can only exist when $\gamma\leq 2$: 
\begin{lemma}[the Moore bound, see, e.g., \cite{Bol13}]
Every $d$-regular graph of girth $g$ on $n$ vertices satisfies $n \ge 2 \cdot \frac{(d-1)^{\frac{g}{2}}-1}{d-2}$.
\end{lemma}
Still, no families with $\gamma$ approaching $2$ for arbitrary girths are known.
The Ramanujan graphs of Lubotzky, Phillips and Sarnak \cite{LPS88} were shown
to obtain $\gamma \ge 4/3$ by Biggs and Boshier \cite{BB90}. 
Lazebnik, Ustimenko and Woldar \cite{LUW95} slightly improved upon
 \cite{LPS88} in the lower-order terms, but more importantly for us,
the family they construct consists of \emph{edge-transitive graphs}. 

\begin{theorem}[\cite{LUW95}]\label{thm:luw}
For every prime power $d$ and even integer $g \ge 6$
there exists an explicit simple, edge-transitive graph with 
$n \le 2d^{g-\left\lfloor \frac{g-3}{4} \right\rfloor -4}$ vertices and girth $g$.
In particular, for every prime power $d$
there exists a $(g,\gamma=4/3)$-Moorish family of edge-transitive $d$-regular graphs,
where $\Image(g) = \set{6,8,\ldots}$.
\end{theorem}

Intuitively, in an edge-transitive graph the local environment
of every edge (i.e., the vertices and edges adjacent to it) looks the same.
More formally, an edge-transitive graph is one in which any two edges are 
equivalent under some element of its automorphism group.
As the
computation of the effective resistance is not affected by an automorphism,
we can conclude the following claim.

\begin{claim}
Let $G = (V,E)$ be an unweighted edge-transitive graph. Then, 
for every two edges $e = (a,b)$ and $e'=(a',b')$ in $E$
it holds that $R_{ab} = R_{a'b'}$.
\end{claim}

\subsection{The Lower Bound Proof}

We next prove our main result for this section, showing 
that Moorish edge-transitive graphs cannot be sparsified
via bounded-independence edge sampling when $k$ is too small. Our proof can be seen as an extension of an argument by
Alon and Nussboim \cite{AN08}, who studied the bounded independence
relaxation of the usual Erd\H{o}s-R\'{e}nyi random graph model, where
it is only required that the distribution of any subset of
$k$ edges is independent. They
provide upper and lower bounds on the minimal $k$ required
to maintain properties that are satisfied by a truly random graph, and in
particular they show that there exists a pairwise independent distribution $\mathcal{D}$ over edges with marginals $1/2$ such that a random graph sampled from $\mathcal{D}$ is disconnected almost surely.

As a warm-up, we extend the argument in \cite{AN08} and show that \emph{3-wise} independence
also does not suffice, even for the special case of sparsifying the complete graph.
\begin{lemma}\label{lem:three-wise}
Let $G = (V = [n],E)$ be the complete graph. There exists a distribution $\mathcal{D} \sim \B^{|E|}$ that
is 3-wise independent with marginals $1/4$ such that sampling a subgraph of $G$ according
to $\mathcal{D}$ would result in a disconnected graph with probability at least $1-2^{1-n}$.
\end{lemma}
\begin{proof}
We first set some notations. Let $\mathcal{G}(A,p)$ be the usual Erd\H{o}s-R\'{e}nyi model, in which
each edge between two vertices in $A$ is included in the graph with probability $p$. Let
$\mathcal{B}(A)$ be the natural distribution over complete bipartite graphs: Choose a partition
$A = A_1 \uplus A_2$ uniformly at random and include all edges between $A_1$ and $A_2$. 

We construct $\mathcal{D} \sim \B^{|E|}$ as follows. Choose a partition $[n] = V_0 \uplus V_1$
uniformly at random. On $V_0$, draw a graph from $\mathcal{G}(V_0,1/2)$ and on
$V_1$, draw a graph from $\mathcal{B}(V_1)$. Clearly, sampling $G'$ according to $\mathcal{D}$
would result in a disconnected graph with probability at least $1-2^{1-n}$, so what is left to
show is that $\mathcal{D}$ is $3$-wise independent with marginals $1/4$.

To see this, let us first consider a single edge $e = (a,b) \in E$. Notice that each edge
in a graph sampled from $\mathcal{B}$ appears with probability $1/2$. Thus,
$$
\Pr[\mathcal{D}(e) = 1] = \Pr[a,b \in V_0] \cdot \frac{1}{2} + \Pr[a,b \in V_1] \cdot \frac{1}{2} = \frac{1}{4}.
$$ 
Next, fix a set $A \subseteq E$ of $t \in \set{2,3}$ edges and note that we can assume without loss of generality
that these edges form either a path or a triangle (for $t=3$), as disjoint paths will occur independently. 
If $A$ forms a path, then similarly,
$$
\Pr[A \in \mathcal{D}] = \Pr[\cup A \subseteq V_0] \cdot 2^{-t} + 
\Pr[\cup A \subseteq V_1] \cdot 2^{-t} = 2 \cdot 2^{-t-1} \cdot 2^{-t} = 4^{-t},
$$
which is what we want. If $A$ forms a triangle, then using the fact that a bipartite graph is triangle-free,
$$
\Pr[A \in \mathcal{D}] = \Pr[\cup A \subseteq V_0] \cdot \frac{1}{8} = 4^{-3},
$$
concluding the proof.
\end{proof}

The above lemma shows that one cannot sparsify the complete graph via $(k=3)$-wise
independent edge sampling.  
For a general $k$, we indeed need to resort to Moore-like graphs.
\begin{proof}[Proof of \cref{thm:main-lower}]
Recalling that $k= \left\lfloor4/3\alpha\right\rfloor$,
let $g = k+1$ or $g = k + 2$, 
whichever is even. Set $d_0$ to be the first prime power larger than
$$1+\max\set{2^{\frac{6^2}{\alpha^8}},(2c)^{\frac{6}{\alpha^2}}}.$$
By \cref{thm:luw}, for every prime power $d \ge d_0$ there exists
$n = n(g,d)$ and a girth-$g$, edge-transitive, $d$-regular graph
graph $G = (V=[n],E)$. Note that by our choice of parameters, indeed $g \ge 6$.
From here onwards, fix such a $d$ and $n = n(g,d)$, observing that $\set{n(g,d)}_{d \ge d_0}$
is infinite.

To get marginals $1/2$,
we will choose $\alpha_0$ so that $c\log n \cdot n^{\alpha_0} = d/2$,
and so we have, using the fact that $d \ge n^{4/3g}+1$,
$$
\alpha_0 \ge \frac{4}{3g} - \frac{\log(2c\log n)}{\log n} \ge \frac{4}{\frac{4}{\alpha}+6} - \frac{\log(2c\log n)}{\log n}
\ge \left( 1 - \frac{3\alpha}{2} \right)\alpha - \frac{\log(2c\log n)}{\log n}.
$$
As $n \le (d-1)^{\frac{3g}{4}}$ and 
$n \ge 2 \cdot \frac{(d-1)^{\frac{g}{2}}-1}{d-2} \ge (d-1)^{\frac{g}{2}-1}$, the latter being the Moore bound, we have
$$
\frac{\log(2c\log n)}{\log n} \le \frac{\log(2c)+\log\frac{3g}{4}+\log\log(d-1)}{\left( \frac{g}{2}-1 \right)\log(d-1)} \le \frac{\log(2c)}{\log(d-1)} +2 \cdot \frac{\log\log(d-1)}{\log(d-1)} \le \frac{\alpha^2}{2},
$$
where we used $\frac{\log(2c)}{\log(d-1)} \le \frac{\alpha^2}{6}$, $\log\frac{3g}{4} \le \frac{g}{2}-1$ and $\frac{\log\log(d-1)}{\log(d-1)} \le \frac{\alpha^2}{6}$.
Thus, overall,
$\alpha_0 \ge (1-2\alpha)\alpha$.

We now give a $k$-wise independent distribution with marginals $1/2$ which fails to yield a good spectral sparsifier for $G$, namely it will be disconnected with high probability.

To do so, construct $\mathcal{D} \sim \set{0,1}^{|E|}$ as follows. Choose
a partition $[n] = V_0 \uplus V_1$ uniformly at random. Each random partition
gives rise to an element $D \sim \mathcal{D}$ in which
for $e = (u,v) \in E$,
 $D(e) = 1$ (i.e.,
the edge $e$ is chosen to survive) if and only if either
$u,v \in V_0$ or $u,v \in V_1$.
\begin{claim}
The distribution $\mathcal{D}$ is $k$-wise independent with marginals $1/2$.
\end{claim}
\begin{proof}
Let $A \subseteq E$ be a set of $t \le k$ edges of $G$. We want
to show that $\Pr[A \in \mathcal{D}] = 2^{-t}$. First, similar to \cref{lem:three-wise}, note that
we can assume without loss of generality that $A$ is a connected component,
since whenever $A_1$ and $A_2$ are over disjoint sets of 
vertices, $\Pr[A_{1} \cup A_{2} \in \mathcal{D}] = \Pr[A_{1} \in \mathcal{D}] \cdot 
\Pr[A_{2} \in \mathcal{D}]$. As the girth of $G$ is larger than
$t$, it must be the case that $A$ is a tree.  

In such a case, where $A$ contains no cycles, $\Pr[A \in \mathcal{D}]$ is equal to the probability that all
$t+1$ vertices in $A$ belong to the same partition, which is $2 \cdot 2^{-(t+1)} = 2^{-t}$.
\end{proof}
By the way $\mathcal{D}$ was constructed, it is clear that 
sampling $G'$ according to $\mathcal{D}$ would result in a disconnected
graph with probability $1-2^{1-n}$, meaning that $G'$ almost surely does
not $\eps$-approximate $G$, for \emph{any} $\eps$. 
\end{proof}

We again stress that by the work in \cref{sec:upper-bound}, we know that \emph{any}
$k$-wise independent distribution over the edges of $G$ with marginals $s\cdot(n-1)/|E| = O\left(d\cdot n/\epsilon^2\cdot |E|\right)$ for $k = \left\lceil 2/\alpha\right\rceil$ would produce an $\epsilon$-spectral sparsifier with expected degree $O(d)$ with high constant probability.

The above also implies that any improvement upon Moorish families of
edge-transitive graphs will improve our lower bound. 
Assuming the existence of a $(g,\gamma = 2)$-Moorish family of 
edge-transitive graphs
we are able to show that the result of \cref{sec:upper-bound} is
essentially tight. Furthermore, if we could take $g$ to be an arbitrary
function of the number of vertices $n$, a similar results would hold
even for non constant $\alpha  = \Omega(1/\log n)$.

%

\section{Spectral Sparsifiers in Deterministic Small Space}\label{sec:algorithm}
In this section we show that $\Sparsify$ can be derandomized space efficiently.

\begin{theorem}[deterministic small-space sparsification]
\label{thm:derand}
Let $G$ be an undirected, connected, weighted graph on $n$ vertices with Laplacian $L$. There is a deterministic algorithm that, when given $G$, an even integer $k$ and $0 < \epsilon < 1$ outputs a weighted graph $H$ with Laplacian $\tilde{L}$ satisfying:
\begin{enumerate}
\item $\tilde{L}\approx_{\epsilon}L$, and,
\item $H$ has $O\left(\frac{\log n}{\eps^{2}} n^{1+2/k}\right)$ edges. 
\end{enumerate}
The algorithm runs in space $O(k \log (N\cdot w)+\log (N\cdot w) \log\log (N\cdot w))$, where $w=\wmax/\wmin$ is the ratio of the maximum and minimum edge weights in $G$ and $N$ is the bitlength of the input.
\end{theorem}

We use the standard model of space-bounded computation. The machine has a read-only input tape, a constant number of  read/write work tapes, and a write-only output tape. We say the machine runs in space $s$ if throughout the computation, it only uses $s$ total tape cells on the work tapes. The machine may write outputs to the output tape that are larger than $s$ (in fact as large as $2^{O(s)}$) but the output tape is write-only. We use the following fact about the composition of space-bounded algorithms.
\begin{lemma}
\label{prop:composition}
Let $f_1$ and $f_2$ be functions that can be computed in space $s_1(n),s_2(n)\geq \log n$, respectively, and $f_2$ has output of length $\ell_1(n)$ on inputs of size $n$. Then $f_2\circ f_1$ can be computed in space 
\[
O(s_2(\ell_1(n)) + s_1(n)).
\]
\end{lemma}

The natural way to derandomize $\Sparsify$ would be to iterate
over all elements of the corresponding $k$-wise independent sample space.
More formally, given $\set{p_{ab}}_{(a,b)\in E}$, let $I_{ab}$ be the indicator
random variable that is $1$ if and only if edge $(a,b)$ is chosen. If the
$I_{ab}$'s are $k$-wise independent so that $\Pr[I_{ab} = 1] = p_{ab}$ (or some good 
approximation of $p_{ab}$), 
we are guaranteed to succeed with nonzero probability. Hence, at least one assignment
to the $I_{ab}$'s taken from the $k$-wise independent is guaranteed to work. From \cref{sec:limited-independence}
we know the sample space is small enough that we can afford to enumerate over all elements in it.
Towards proving \cref{thm:derand}, there are still three issues to consider:
\begin{enumerate}
	\item Approximating the effective resistances $R_{ab}$ for every $(a,b) \in E$, space efficiently.
	Fortunately, we can do this with high accuracy using the result of Murtagh, Reingold, Sidford, and Vadhan \cite{MRSV17} for
	approximating the pseudoinverse of a Laplacian, which we state shortly.
	\item Verifying that a given set of random choices in $\Sparsify$ provides a sparse and accurate approximation to the input graph. The sparsity requirement is easy to check. To check that $\tilde{L} \approx_{\epsilon} L$, we devise a verification algorithm that uses the algorithm of \cite{MRSV17}. The details are given in \cref{lem:verification}.
	\item The Laplacian solver of \cite{MRSV17} only works for multigraphs (graphs with integer edge weights) and we want an algorithm that works for general weighted graphs. To fix this, we extend the work of \cite{MRSV17} by giving a simple reduction from the weighted case to the multigraph case. The details can be found in \cref{app:weighted_solver}.
\end{enumerate}

\subsection{Algorithm for Approximating Effective Resistances}
\label{sect:approx_er_alg}
As noted above, a key ingredient in our deterministic sparsification algorithm is a deterministic nearly logarithmic space algorithm for approximating the pseudoinverse of an undirected Laplacian. 

\begin{theorem}[\cite{MRSV17}]
\label{lem:solver}
Given an undirected, connected multigraph $G$ with Laplacian $L=D-A$ and $\epsilon>0$, there is a deterministic algorithm that computes a symmetric PSD matrix $\tilde{L^+}$ such that $\tilde{L^+}\approx_{\epsilon}L^{+}$, and uses space $O(\log N\cdot\log\log\frac{N}{\eps})$, where $N$ is the bitlength of the input (as a list of edges).
\end{theorem}
Note that the space complexity above assumes that the multigraph is given as a list of edges. If we instead think of parallel edges as integer edge weights, then $N$ should be replaced by $N\cdot \wmax$ where $\wmax$ is the maximum edge weight in $G$ since an edge of weight $w$ gets repeated $w$ times in the edge-list representation. To work with general weighted graphs, we extend the result of \cite{MRSV17}.

\begin{lemma}[small space laplacian solver for weighted graphs]
\label{lem:weighted_solver}
Given an undirected connected weighted graph $G=(V,E,w)$ with Laplacian $L=D-A$, and $0 < \epsilon < 1$, there exists a deterministic algorithm that computes a symmetric PSD matrix $\tilde{L^+}$ such that $\tilde{L^+}\approx_{\epsilon}L^{+}$, and uses space $O(\log (N\cdot w)\log\log (N\cdot w/\eps))$, where $w=\wmax/\wmin$ is the ratio of the maximum and minimum edge weights in $G$ and $N$ is the bitlength of the input.
\end{lemma}

A proof of \cref{lem:weighted_solver} can be found in \cref{app:weighted_solver}. \cref{lem:weighted_solver}  immediately gives an algorithm for computing strong multiplicative approximations to effective resistances.
\begin{lemma}
\label{lem:comp_ers}
Let $G=(V,E,w)$ be an undirected, connected, weighted graph and let $R_{ab}$ be the effective resistance of $(a,b)\in E$. There is an algorithm that computes a real number $\tilde{R}_{ab}$ such that 
\[
(1-\epsilon)\cdot R_{ab} \leq\tilde{R}_{ab}\leq(1+\epsilon)\cdot R_{ab}
\]
and uses space $O(\log (N\cdot w)\cdot\log\log\frac{N\cdot w}{\eps})$, where $w=\wmax/\wmin$ is the ratio of the maximum and minimum edge weights in $G$ and $N$ is the bitlength of the input.
\end{lemma}
\begin{proof}
Let $L$ be the Laplacian of $G$.
By the definition of effective resistance we have $R_{ab}=(e_a-e_b)^{\top} L^+(e_a-e_b)$. From \cref{lem:weighted_solver} we can compute a matrix $\tilde{L}$ such that 
\[
(1-\epsilon)\cdot L^+ \preceq\tilde{L^+}\preceq(1+\epsilon)\cdot L^+
\]
in space $O(\log (N\cdot w)\cdot\log\log\frac{N\cdot w}{\eps})$. By the definition $\preceq$, this implies
\[
(1-\epsilon)\cdot (e_a-e_b)^{\top} L^+(e_a-e_b) \leq(e_a-e_b)^{\top}\tilde{L^+}(e_a-e_b)\leq(1+\epsilon)\cdot (e_a-e_b)^{\top} L^+(e_a-e_b).
\]
Setting $\tilde{R}_{ab}=(e_a-e_b)^{\top}\tilde{L^+}(e_a-e_b)$ and noting that the vector matrix multiplication only adds logarithmic space overhead completes the proof. 
\end{proof}

\subsection{Testing for Spectral Proximity}\label{sec:verify}
In this section we give our deterministic, small-space procedure for verifying that two Laplacians spectrally approximate one another. We will need the following claim about the space complexity of matrix multiplication.

\begin{claim}
\label{lem:matrixprod}
Given $n\times n$ matrices $M_1,\ldots,M_k$, their product $M_1\cdot\ldots\cdot M_k$ can be computed using $O(\log N\cdot \log k)$ space, where $N$ is the bitlength of $(M_{1},\ldots,M_k)$. 
\end{claim}
The proof of \cref{lem:matrixprod} uses the natural divide and conquer algorithm and the fact that two matrices can be multiplied in logarithmic space. A detailed proof can be found in \cite{MRSV17}.

Using \cref{lem:weighted_solver} and \cref{lem:matrixprod}, we prove the following lemma. The high level idea is that testing whether two matrices $L$ and $\tilde{L}$ spectrally approximate each other can be reduced to approximating the spectral radius of a particular matrix 
$$
M = \left( \frac{(\tilde{L}-L)L^{+}}{\eps} \right)^{2}.
$$
In fact, it will be sufficient to check whether the trace of a sufficiently high power of $M$
is below a certain threshold to deduce whether the spectral radius of $M$ does not exceed $1$. For intuition, replace the matrices with scalars $m,\ell,$ and $\tilde{\ell}$ where
 \[
m=\frac{(\tilde{\ell}-\ell)^2}{(\eps\cdot\ell)^2}.
\]
Then, $m\leq 1$ implies $\sqrt{m}\leq 1$, which implies $|\tilde{\ell}-\ell|\leq \eps\cdot\ell$ -- the kind of relative closeness we want between the matrices $\tilde{L}$ and $L$ when aiming for spectral approximation.

\begin{lemma}
\label{lem:verification}
There exists a deterministic algorithm that, given undirected, connected, weighted graphs $\tilde{G}$ and $G$ with Laplacians $\tilde{L},L$, and $\epsilon,\alpha>0$, outputs \texttt{YES} or \texttt{NO} such that
\begin{enumerate}
    \item  If $\tilde{L}\approx_{\epsilon}L$, then the algorithm outputs \texttt{YES}, and,

    \item  If $\tilde{L}\not\approx_{\epsilon\cdot\sqrt{1+\alpha}}L$ then the algorithm outputs \texttt{NO}.
    \end{enumerate}

The algorithm uses space $O(\log (N\cdot w)\cdot \log\log \frac{N\cdot w}{\alpha\eps}+\log (N\cdot w) \cdot \log\frac{1}{\alpha})$, where $w=\wmax/\wmin$ is the ratio of the maximum and minimum edge weights in $G$ and $\tilde{G}$ and $N$ is the bitlength of the input.
\end{lemma}

\begin{proof}
Let 
$$
M = \left( \frac{(\tilde{L}-L)L^{+}}{\eps} \right)^{2}.
$$ 
Set $T = \Tr(M)$ and $t = \left\lceil \frac{\log T}{\log(1+\alpha)} \right\rceil$. The following claim shows that if we can compute $\Tr(M^t)$ exactly then we can check the two cases in \cref{lem:verification}. However we won't be able to compute $\Tr(M^t)$ exactly because that would require computing $L^+$ exactly. This will be addressed later. 

\begin{claim}
\label{claim:trace_alg}
If $\tilde{L}\approx_{\epsilon}L$ then $\Tr(M^t)\leq T$ and if $\tilde{L}\not\approx_{\epsilon\cdot\sqrt{1+\alpha}}L$ then $\Tr(M^t) > T$.
\end{claim}

\begin{proof}
Let $\Pi=I-J$ be the orthogonal projection onto $\Image(L) = \Image(\tilde{L})$ (note
that both $L$ and $\tilde{L}$ are the Laplacians of connected graphs). Using \cref{lem:proj_matrix} and \cref{lem:mult_both_sides}, we know that $\tilde{L}\approx_{\epsilon}L$ if and only if for all $v\in\R^n$ we have
\[
-v^{\top} \Pi v \leq \frac{1}{\epsilon}\cdot v^{\top} L^{+/2}\left(\tilde{L}-L\right)L^{+/2}v \leq v^{\top} \Pi v.
\]
Note that $\Pi$ and $L^{+/2}\left(\tilde{L}-L\right)L^{+/2}$ have the same kernel, namely $\Span(\set{\onesvec})$, and being perpendicular to $\onesvec$ is preserved under both operators. Thus, the above holds if and only if it holds on all vectors $v\perp \onesvec$. For such vectors we have $\Pi v=v$ and hence $v^{\top}\Pi v=\|v\|^2$. So we have $\tilde{L}\approx_{\epsilon}L$ if and only if for all vectors $v\perp \onesvec$ we have  
\[
\left|\frac{1}{\epsilon}\cdot v^{\top} L^{+/2}\left(\tilde{L}-L\right)L^{+/2}v\right|\leq \|v\|^2,
\]
or equivalently,
 
\[
\lm\left(\frac{1}{\epsilon}\cdot L^{+/2}\left(\tilde{L}-L\right)L^{+/2}\right)\leq 1.
\]

Note that $L^{+/2}(\tilde{L}-L)L^{+/2}$ is symmetric so has real eigenvalues and is similar to the matrix $(\tilde{L}-L)L^{+}$ on the space orthogonal to the kernel. Thus, we can rewrite the above condition as 
\[
\lm\left(\frac{1}{\epsilon}\cdot\left(\tilde{L}-L\right)L^{+}\right)\leq 1.
\]
Furthermore, we have that $M=((\tilde{L}-L)L^+/\epsilon)^2$ has real, non-negative eigenvalues. 

Consider the contrapositives of the implications stated in the claim. Namely: If $\Tr(M^t)>T$ then $\tilde{L}\not\approx_{\epsilon}L$ and if
$\Tr(M^t)\leq T$ then $\tilde{L}\approx_{\epsilon\cdot\sqrt{1+\alpha}}L$. Now, note that if $\Tr(M^t)>T = \Tr(M)$ then $\lm(M)>1$ because the only way that the trace of a matrix with real non-negative eigenvalues can increase under powering is if at least one of its eigenvalues exceeds 1. So, on the one hand, if $\Tr(M^t)>T = \Tr(M)$
then $\tilde{L} \not\approx_{\eps} L$. On the other hand, 
\begin{align*}
    \Tr(M^t)\leq T &\implies \lm(M)^t \leq T\\
    &\implies \rho(M)\leq T^{1/t}\leq (1+\alpha).
\end{align*}
When $\lm(M)\leq 1+\alpha$, we have $\rho(L^{+/2}(\tilde{L}-L)L^{+/2}/\eps)\leq\sqrt{1+\alpha}$, so
\[
-\sqrt{1+\alpha}\cdot\Pi \preceq \frac{1}{\epsilon}\cdot L^{+/2}(\tilde{L}-L)L^{+/2}\preceq \sqrt{1+\alpha}\cdot\Pi, 
\]
or in other words, $\tilde{L}\approx_{\epsilon\cdot\sqrt{1+\alpha}}L$.
\end{proof}
Recall that 
we cannot compute $M$ exactly in small space because we do not know how to compute $L^+$ exactly in small space. Let 
$$\hat{M} = \left( \frac{1}{1+\gamma}\cdot\frac{(\tilde{L}-L)\hat{L}}{\eps} \right)^{2},$$ where $\hat{L} \approx_{\gamma} L^{+}$ for 
\[
\gamma=1-\frac{2}{1+\sqrt{1+\frac{\alpha}{2+\alpha}}}.
\]
We have chosen $\gamma$ so that 
$
\left(\frac{1+\gamma}{1-\gamma}\right)^2 = 1+\frac{\alpha}{2+\alpha}
$,
and hence 
\[
\left(\frac{1+\gamma}{1-\gamma}\right)^2\cdot\left(1+\frac{\alpha}{2}\right) = 1+\alpha.
\]
Now we will show that $\hat{M}$ is sufficient for our purposes. For this,
we will need the following claim, whose proof is deferred to \cref{proofs-appendix}.
\begin{claim}
\label{lem:norm_of_prod_ineq}
Let $C$ be a real, symmetric matrix and let $A$ and $B$ be real, symmetric, PSD matrices such that $A\preceq B$. Suppose $\ker(A)=\ker(B)=\ker(C)$. Then, 
\[
\lm(CA)\leq \lm(CB).
\]
\end{claim}
Next, we show that $\hat{M}$ is sufficient for our purposes.

\begin{claim}\label{claim:dist}
It holds that
$\left(\frac{1-\gamma}{1+\gamma}\right)^2\cdot \lm(M) \leq \lm(\hat{M}) \leq \lm(M)$.
\end{claim}
This claim is sufficient because we can use the procedure implied by \cref{claim:trace_alg} to distinguish the case of $\rho(\hat{M}) \leq1$ from the case of $\rho(\hat{M})>1+\alpha'$, for $\alpha' = \frac{\alpha}{2}$.
From \cref{claim:dist}, we know that
\[
\rho(M)\leq 1 \implies \rho(\hat{M})\leq 1
\]
and that
\[
\rho(M)>1+\alpha\implies\rho(\hat{M})>\left(\frac{1-\gamma}{1+\gamma}\right)^2\cdot (1+\alpha)=1+\alpha'
\]
Thus, from the arguments above, we can distinguish the case of $\tilde{L}\approx_{\epsilon} L$ from $\tilde{L}\not\approx_{\epsilon\cdot\sqrt{1+\alpha}} L$ by computing $\Tr(\hat{M}^{\hat{t}})$ for $\hat{t}=\left\lceil{\log\Tr(\hat{M})/\log(1+\alpha')}\right\rceil$ and comparing the result to $\Tr(\hat{M})$. 
Now we prove \cref{claim:dist}.
\begin{proof}
By assumption we have $\hat{L}\approx_{\gamma}L^+$ and hence
\[
\frac{1-\gamma}{1+\gamma}\cdot L^+\preceq \frac{1}{1+\gamma}\cdot\hat{L}\preceq L^+.
\]
Let $C=\frac{1}{\eps}(\tilde{L}-L)$ and note that $C$ is symmetric. By \cref{lem:norm_of_prod_ineq}, we have that 
\[
\frac{1-\gamma}{1+\gamma}\cdot \rho(CL^+)\leq \frac{1}{1+\gamma}\cdot \rho(C\hat{L})\leq \rho(CL^+),
\]
and hence that 
\[
\left(\frac{1-\gamma}{1+\gamma}\cdot \rho(CL^+)\right)^2\leq \left(\frac{1}{1+\gamma}\cdot \rho(C\hat{L})\right)^2\leq \rho(CL^+)^2.
\]
Noting that for all matrices $A$ with real eigenvalues, we have $\rho(A^2)=\rho(A)^2$ and that $(CL^+)^2=M$ and $(C\hat{L}/(1+\gamma))^2=\hat{M}$, the above becomes 
\[
\left(\frac{1-\gamma}{1+\gamma}\right)^2\cdot \rho(M)\leq \rho(\hat{M})\leq \rho(M),
\]
as desired.
\end{proof}

Thus, our distinguishing algorithm goes as follows: Approximate
$\hat{L} \approx_{\gamma} L^{+}$ using the Laplacian solver algorithm given in \cref{lem:weighted_solver},
compute $\Tr(\hat{M}^{\hat{t}})$ and answer according to whether it is greater than or less than $\Tr(\hat{M})$. \cref{claim:trace_alg} and \cref{claim:dist} establishes
its correctness. We are left with establishing the space complexity. 

\begin{claim}
The distinguishing algorithm uses space $O(\log (N\cdot w)\cdot \log\log \frac{N\cdot w}{\alpha\eps}+\log (N\cdot w) \cdot \log\frac{1}{\alpha})$, where $w=\wmax/\wmin$ is the ratio of the maximum and minimum edge weights in $G$ and $\tilde{G}$ and $N$ is the bitlength of the input.
\end{claim}
\begin{proof}
By \cref{lem:weighted_solver}, we can compute $\hat{L}$ in space $S=O(\log (N\cdot w)\log\log\frac{N\cdot w}{\gamma})=O(\log (N\cdot w)\cdot\log\log\frac{N\cdot w}{\alpha})$.  \cref{lem:matrixprod} and composition of space-bounded algorithms (\cref{prop:composition}) say that we can compute $\hat{M}^{\hat{t}}$ using additional 
\[
S'=O\left(S\cdot \log \hat{t}\right)=O\left(\log (N\cdot w)\cdot \left(\log\log \Tr(\hat{M})+\log\frac{1}{\log(1+\alpha)}\right)\right)
\] 
space. Finally, computing $\Tr(\hat{M}^{\hat{t}})$ only takes an addition $O(S')$ space to add $n$ numbers.

Now we argue that we can (loosely) bound $\Tr(\hat{M})$ by $\poly(N\cdot w,1/\alpha,1/\eps)$.  
If $\dmax$ is the maximum weighted degree of the graph corresponding to $L$ and $\tilde{L}$ then we have that the spectral norms of $L$ and $\tilde{L}$ are at most $2\cdot \dmax=O(N\cdot\wmax)$. 
\cref{claim:norm_of_pinv_stversion} says that the smallest nonzero eigenvalue of $L$ is lower bounded by $w_{\mathrm{min}}/n^2$. Note that $\|L^+\|$ equals the reciprocal of the smallest nonzero eigenvalue of $L$ and hence  we have $\|L^+\|\leq n^2/w_\mathrm{min}$ and therefore  $\|\hat{L}\|\leq (1+\gamma)\cdot n^2/w_\mathrm{min}=\poly(N/\wmin,\frac{1}{\alpha})$. It follows that $\Tr(\hat{M})=\poly(N\cdot w,1/\alpha,1/\epsilon)$. Plugging this into the space complexity gives us a space bound of
\[
O\left(\log (N\cdot w)\cdot \left(\log\log \frac{N\cdot w}{\alpha \epsilon}+\log\frac{1}{\log(1+\alpha)}\right)\right).
\]
Noting that $\log\frac{1}{\log(1+\alpha)}=O(\log\frac{1}{\alpha})$ completes the analysis.
\end{proof}
\end{proof}

\subsection{Completing the Proof of \cref{thm:derand}}

We can now prove \cref{thm:derand}. As noted above, the algorithm proceeds 
by first approximating the sampling probabilities and then
sparsifying $G$ where the surviving edges are chosen from a small $k$-wise independent
sample space whose marginals are set properly. Each potential sparsifier
is checked using the algorithm given in \cref{sec:verify}.

\begin{proof}[Proof of \cref{thm:derand}]
Set $\delta = \frac{1}{4}$, $\hat{\epsilon}=\frac{4\epsilon}{5}$ and
$$
s= \frac{18 e\log n}{\hat{\eps}^2} \cdot \left(\frac{n}{\delta}\right)^{2/k},
$$ 
for $\alpha$ soon to be determined.
These parameters are chosen in accordance with the parameters
required for $\Sparsify$ to succeed with probability $1/2$
and approximation error $\hat{\eps}$ (see \cref{lem:approx_er}). Set $\alpha' = \alpha/(4+\alpha)$.
We compute approximate effective resistances $\tilde{R}_{ab}$ for each edge $(a,b)$ in $G$ using \cref{lem:comp_ers}, so that 
$$
(1-\alpha')R_{ab} \le \tilde{R}_{ab} \le (1+\alpha')R_{ab}.
$$
This takes $O(\log (N\cdot w) \log\log((N\cdot w)/\alpha))$ space. Then, we compute approximate sampling probabilities as follows: 
$$\tilde{p}_{ab}=\alpha' \cdot\left\lfloor{\frac{1}{\alpha'}\cdot \min\set{1,w_{ab}\cdot\tilde{R}_{ab}\cdot s/(1-\alpha')}}\right\rfloor$$ 
That is, we truncate the required (approximate) sampling probabilities 
to $\log\frac{1}{\alpha'}$ bits of precision.
In particular, denoting the precise sampling probabilities by
$p_{ab}^\star = \min\set{1,w_{ab} \cdot R_{ab} \cdot s}$, we
have 
\begin{align*}
\min\{1,w_{ab}\cdot\tilde{R}_{ab}\cdot s/(1-\alpha')\}-p_{ab}^* &\leq w_{ab}\cdot s\cdot R_{ab}\cdot\left(\frac{1+\alpha'}{1-\alpha'}-1\right)\\
&=p_{ab}^*\cdot \frac{2\alpha'}{1-\alpha'}\\
&\leq \alpha/2
\end{align*}
Furthermore, we have an additional error of $\alpha'$ due to the truncation so  $\left|\tilde{p}_{ab}-p_{ab}^\star\right| \le \alpha/2+\alpha'\leq \alpha$.

We want to set $\alpha$ so that $\tilde{p}_{ab}$ is a multiplicative approximation to $p^*_{ab}$ for all $(a,b)\in E$, which requires $\alpha$ to be smaller than $\min_{(a,b)\in E)}\{p^*_{ab}\}$. 
\begin{claim}
Let $\dmax$ be the maximum weighted degree over all vertices in $G$. Then, for all $(a,b)\in E$, $p^{\star}_{ab}\geq 1/\dmax$.
\end{claim}
\begin{proof}
Since $s>1$ and $w_{ab}\geq 1$ (all edge weights are positive integers) we have $p^{\star}_{ab}\geq R_{ab}$. Let $\lambda_{\mathrm{min}}(C)$ denote the minimal nonzero eigenvalue of a matrix $C$. To lower bound $R_{ab}$, we use the variational characterization of eigenvalues and the definition of effective resistance to write
\begin{align*}
R_{ab}&=(e_a-e_b)^{\top}L^+(e_a-e_b)\\
&\geq \lambda_{\mathrm{min}}(L^+)\cdot \|e_a-e_b\|^2\\
&=\frac{2}{\|L\|}\\
&\geq \frac{1}{\dmax}.
\end{align*}
Note that we can indeed consider the minimal nonzero eigenvalue of $L^+$ because $e_a-e_b$ is perpendicular to the one-dimensional kernel of $L$ (the all-ones vector).
\end{proof}
 In light of the above, we can set $\alpha$ so that $1/\alpha = 2 \cdot \dmax=O(N\cdot w)$ and get a $1/2$-multiplicative approximation to the sampling probabilities. 

Now, consider the $k$-wise independent sample space $\mathcal{D} \subseteq \B^{|E|}$
guaranteed to us by \cref{lem:sampling}, substituting $t = \lceil{\log(1/\alpha')}\rceil$.
By \cref{lem:sampling}, each element of $\mathcal{D}$ can be sampled using
$$O(k\cdot\max\{\log(1/\alpha'),\log |E|\})=O(k\cdot \log (N\cdot w))$$ space.
For each element of $\mathcal{D}$, construct the corresponding sparse graph.
Note that the space used to cycle through each element can be reused.
\cref{lem:approx_er} tells us that at least $1-2\delta=1/2$ of the Laplacians of the resulting graphs $\hat{\eps}$-approximate the Laplacian of $G$ and have
\[
O\left(\frac{1+1/2}{1-1/2} \cdot \frac{1}{\delta^{1+2/k}} \cdot \frac{\log n}{\hat{\eps}^{2}}\cdot n^{1+\frac{2}{k}}\right)=O\left(\frac{\log n}{\eps^{2}}\cdot n^{1+\frac{2}{k}}\right)
\]
edges. For each of these graphs, we run the verification algorithm with accuracy parameter $9/16$, which is guaranteed to find a graph with the above sparsity whose Laplacian approximates the Laplacian of $G$ with error
\[
\hat{\eps}\cdot\sqrt{1+\frac{9}{16}}=\frac{4\eps}{5}\cdot\sqrt{\frac{25}{16}}=\eps
\]
in space 
\[
O\left(\log (N\cdot w) \log\log \frac{16N\cdot w}{9\hat{\eps}}+\log (N\cdot w) \log\frac{16}{9}\right)=O\left(\log (N\cdot w) \log\log \frac{N\cdot w}{\eps}\right).
\]
Again, the space used for the verification process can be reused.
Adding up the space complexities gives us a total of 
\[
O\left(k \log (N\cdot w) +\log (N\cdot w) \log\log \frac{N\cdot w}{\eps}\right)
\]
space. Note that the final result is vacuous when $\eps \le 1/n$ so we can without loss of generality assume that $\epsilon \ge 1/n$. This gives a total space complexity of $O(k \log (N\cdot w) +\log (N\cdot w)\log\log (N\cdot w))$.

\end{proof}

\section{Acknowledgments}
We thank Jelani Nelson for his insights on spectral sparsification via $k$-wise independent sampling. We also thank Jaros\l{}aw B\l{}asiok for helpful discussions about random matrices. The first author would like to thank Tselil Schramm and Amnon Ta-Shma for interesting conversations.

\bibliographystyle{alpha}
\bibliography{references}
\appendix
\section{Deferred Proofs} \label{proofs-appendix}

\begin{namedtheorem}[\cref{lem:norm_of_sum} restated]
Let $A,B, C$ be $n \times n$ symmetric PSD matrices and suppose that $B\preceq C$. Then,
\[
\|A+B\|\leq \|A+C\|.
\]
\end{namedtheorem}
\begin{proof}
For every symmetric matrix $M$ we have 
$\norm{M}= \max_{x : \norm{x}=1}x^{\top}Mx$. Thus, 
$$
\norm{A+B} = \max_{x : \norm{x}=1}\left(x^{\top}Ax+x^{\top}Bx \right)\le 
\max_{x : \norm{x}=1}\left(x^{\top}Ax+x^{\top}Cx \right)= \norm{A+C}.
$$
\end{proof}

\begin{namedtheorem}[\cref{lem:norm_is_er} restated]
Let $G=(V,E,w)$ be an undirected weighted graph on $n$ vertices with Laplacian $L$. Fix $(a,b)\in E$ and recall that $L_{ab}=(e_a-e_b)(e_a-e_b)^{\top}$. Then,
\[
\left\|L^{+/2}L_{ab}L^{+/2}\right\| = R_{ab}.
\]
\end{namedtheorem}
\begin{proof}
$L_{ab}$ is PSD and has rank 1 so it follows that $L^{+/2}L_{ab}L^{+/2}$ is PSD and has rank 1. Since the trace is the sum of the eigenvalues and the norm is the maximum eigenvalue, this implies that $\|L^{+/2}L_{ab}L^{+/2}\|=\Tr(L^{+/2}L_{ab}L^{+/2})$. Finally, the trace is invariant under cyclic permutations so we have
\begin{align*}
    \Tr\left(L^{+/2}(e_a-e_b)(e_a-e_b)^{\top}L^{+/2}\right)&=\Tr\left((e_a-e_b)^{\top}L^+(e_a-e_b)\right)\\
    &=R_{ab}.
\end{align*}
\end{proof}

\begin{namedtheorem}[\cref{lem:proj_matrix} restated]
Let $G = (V,E,w)$ be an undirected connected weighted graph on $n$ vertices with Laplacian $L$. Let $J$ be the $n\times n$ matrix with $1/n$ in every entry and define $\Pi=I-J$. Then, we have that
\[
\Pi=LL^+=L^+L=L^{+/2}LL^{+/2},
\]
where $\Pi$ is the orthogonal projection onto the image of $L$ (the space orthogonal to the all-ones vector).
\end{namedtheorem}

\begin{proof}
It is clear that for all $v = \alpha \onesvec$ where $\alpha \in \mathbb{R}$ we have $\Pi v=v-Jv=\zerosvec$ and for all $v\perp \onesvec$ we have $\Pi v=v$. So $\Pi$ is the orthogonal projection onto the space orthogonal to the all ones vector. Since $G$ is connected, $\Ker(L)=\Span(\set{\onesvec})$ and for every symmetric matrix $M$, $\Image(M)=\Ker(M)^{\perp}$, it follows that $\Pi$ is also the orthogonal projection onto the image of $L$. 

By the definition of the pseudoinverse, we have that $L$ and $L^+$ have identical eigenvectors. Letting $v_1 = \onesvec,v_2,\ldots,v_n$ be their orthonormal basis of eigenvectors, $\lambda_1 = 0,\ldots,\lambda_n$ be the eigenvalues of $L$ and $0,\frac{1}{\lambda_2},\ldots,\frac{1}{\lambda_n}$ be the eigenvalues of $L^+$, we have that 
\[
LL^+v_1=\zerosvec
\]
and for each $i\in\set{2,\ldots,n}$,
\[
LL^+v_i= \frac{1}{\lambda_i}Lv_i=\frac{\lambda_i}{\lambda_i}v_i=v_i.
\]
Thus, $LL^+$ and $\Pi$ are identical on a set of basis vectors and hence are the same matrix. Similar calculations yield the same result for $L^+L$ and $L^{+/2}LL^{+/2}$.
\end{proof}

\begin{namedtheorem}[\cref{lem:mult_both_sides} restated]
Let $A,B, C$ be symmetric $n\times n$ matrices and suppose $A$ and $B$ are PSD. Then the following hold 
\begin{enumerate}
\item $A\approx_{\eps} B\implies C^{\top}AC\approx_{\eps}C^{\top}BC$
\item If $\ker(C)\subseteq\ker(A)=\ker(B)$ then $A\approx_{\eps} B\iff C^{\top}AC\approx_{\eps}C^{\top}BC$
\end{enumerate}
\end{namedtheorem}
\begin{proof}
First we prove item 1. Note that since $A\approx_{\eps} B$, we have that $\ker(A)=\ker(B)$, otherwise the spectral inequalities can be violated by a vector in the kernel of one of the matrices but not in the kernel of the other.  Fix $x\in\R^n$ and let $y=Cx$. By assumption we have 
\[
(1-\eps)\cdot y^{\top}By\leq  y^{\top}Ay\leq (1+\eps)\cdot y^{\top}By.
\]
Observing that $y^{\top}Ay= x^{\top}C^{\top}ACx$ and $y^{\top}By= x^{\top}C^{\top}BCx$ and noting that $x$ was arbitrary completes the proof. 

For item 2, assume $\ker(C)\subseteq\ker(A)=\ker(B)$ and we will show
\[
 C^{\top}AC\approx_{\eps}C^{\top}BC\implies A\approx_{\eps} B.
 \]
The other direction follows from item 1.  Fix $x\in\R^n$. We want to show that 
\begin{equation}
\label{eq:mult_on_both_sides}
(1-\eps)\cdot x^{\top}Bx\leq  x^{\top}Ax\leq (1+\eps)\cdot x^{\top}Bx.
\end{equation}
If $x\in\ker(A)=\ker(B)$ then the above is trivially true. So without loss of generality, we can take $x\in \ker(A)^{\perp}=\Image(A)=\Image(B)$. Let $y=C^+x$. By assumption we have
\[
(1-\eps)\cdot y^{\top}C^{\top}BCy \leq y^{\top}C^{\top}ACy\leq (1+\eps)\cdot y^{\top}C^{\top}BCy
\]
Let $\Pi=CC^+$. We can rewrite the above as 
\[
(1-\eps)\cdot x^{\top}\Pi^{\top}B\Pi x \leq x^{\top}\Pi^{\top}A\Pi x\leq (1+\eps)\cdot x^{\top}\Pi^{\top}B\Pi x
\]
Note that by the definition of the Moore-Penrose pseudoinverse, $\Pi$ is the projection onto $\ker(C)^{\perp}\supseteq \Image(A)=\Image(B)$. Since we assumed without loss of generality that $x\in\Image(A)=\Image(B)$, it follows that $\Pi x=x$. Substituting this into the above establishes \cref{eq:mult_on_both_sides} and completes the proof.
\end{proof}

\begin{namedtheorem}[\cref{lem:cc_spec_approx} restated]
Let $G$ and $\tilde{G}$ be undirected graphs on $n$ vertices with Laplacians $L$ and $\tilde{L}$, respectively. If $G$ is connected and $\tilde{G}$ is disconnected then $\tilde{L}\not\approx_{\epsilon} L$ for any $\epsilon>0$.
\end{namedtheorem}
\begin{proof}
Since $G$ is a connected, undirected graph we have that $\ker(L)=\Span(\set{\onesvec})$. We will show that there is a vector $v\in\ker(\tilde{L})$ such that $v\not\in\ker(L)$. This will complete the proof because it implies $v^{\top}\tilde{L} v=0$ but $v^{\top}L v\neq 0$ and hence the quadratic forms of these laplacians cannot multiplicatively approximate each other. 

Let $C$ be a connected component of $\tilde{G}$ and let $V\setminus C$ be the remaining vertices. By assumption $V\setminus C\neq \emptyset$. Let $\tilde{L}(C)$ be the Laplacian of $\tilde{G}$ if all of the edges in $V\setminus C$ were deleted and define $\tilde{L}(V\setminus C)$ analogously. Note that $\tilde{L}=\tilde{L}(C)+\tilde{L}(V\setminus C)$.

Define $v$ so that $v_i=0$ for all $i\in V\setminus C$ and $v_i=1$ for all $i\in C$. Then we have that 
\[
v^{\top}\tilde{L}v = v^{\top} \tilde{L}(C)v+v^{\top}\tilde{L}(V\setminus C)v = 0,
\]
but $v\not\in\ker(L)$.
\end{proof}

\begin{namedtheorem}[\cref{lem:norm_of_prod_ineq} restated]
Let $C$ be a real, symmetric matrix and let $A$ and $B$ be real, symmetric, PSD matrices such that $A\preceq B$. Suppose $\ker(A)=\ker(B)=\ker(C)$. Then, 
\[
\lm(CA)\leq \lm(CB).
\]
\end{namedtheorem}
\begin{proof}
We will show that $\lm(A^{1/2}CA^{1/2})\leq\lm(B^{1/2}CB^{1/2})$. This will complete the proof because $A^{1/2}CA^{1/2}$ is similar to $CA$ on the space orthogonal to $\ker(A)=\ker(C)$ (multiply on the left by $A^{+/2}$ and on the right by $A^{1/2}$). Therefore, $A^{1/2}CA^{1/2}$ and $CA$ have the same spectrum and likewise that $B^{1/2}CB^{1/2}$ and $CB$ have the same spectrum. So establishing the inequality would show that $\lm(CA)\leq\lm(CB)$ as desired. To show $\lm(A^{1/2}CA^{1/2})\leq\lm(B^{1/2}CB^{1/2})$ we prove the following claim.

\begin{claim}\label{claim:a1}
Let $M\in\mathbb{R}^{n\times n}$ be a (possibly asymmetric) matrix and let $A,B\in\mathbb{R}^{n\times n}$ be symmetric PSD matrices such that $A\preceq B$. Then,
\begin{enumerate}
    \item $\|A^{1/2}M\|\leq\|B^{1/2}M\|$,
    \item $\|MA^{1/2}\|\leq\|MB^{1/2}\|$.
\end{enumerate}
\end{claim}
\begin{proof}
To prove item (1), we write
 \begin{align*}
  \|A^{1/2}M\|&=\max_{x\colon\|x\|=1}\|A^{1/2}Mx\| \\
  &=\max_{x\colon\|x\|=1}\sqrt{x^{\top}M^{\top}AMx}\\
  &\leq \max_{x\colon\|x\|=1}\sqrt{x^{\top}M^{\top}BMx}\\
  &= \|B^{1/2}M\|,
\end{align*}
where the inequality follows from the assumption that $A\preceq B$. For the second item, first use item (1) to get $\|A^{1/2}M^{\top}\|\leq \|B^{1/2}M^{\top}\|$. Recall that for all matrices $N$, we have $\|N\|=\|N^{\top}\|$. Therefore, since $A$ and $B$ are symmetric. we get $\|MA^{1/2}\|\leq \|MB^{1/2}\|$.
\end{proof}
Now set $M'=A^{1/2}C$ and apply item (2) of \cref{claim:a1} to conclude that $\|A^{1/2}CA^{1/2}\|\leq\|A^{1/2}CB^{1/2}\|$. Set $M''=CB^{1/2}$ and apply item (1) of the same claim to get $\|A^{1/2}CB^{1/2}\|\leq \|B^{1/2}CB^{1/2}\|$. It follows that $\|A^{1/2}CA^{1/2}\|\leq\|B^{1/2}CB^{1/2}\|$. The claim follows from the fact that $A^{1/2}CA^{1/2}$ and $B^{1/2}CB^{1/2}$ are symmetric matrices and for all symmetric matrices $N$, we have $\|N\|=\lm(N)$. 
\end{proof}

\section{Small Space Laplacian Solver for Weighted Graphs}
\label{app:weighted_solver}
In this section we show how to compute an approximate pseudoninverse of a Laplacian of a weighted undirected graph deterministically in small space, extending the work of \cite{MRSV17} who only gave an algorithm for undirected multigraphs (graphs with integer edge weights). We sketch the proof of the following lemma.

\begin{namedtheorem}[\cref{lem:weighted_solver} restated]
Given an undirected connected weighted graph $G=(V,E,w)$ with Laplacian $L=D-A$, and $0 < \epsilon < 1$, there exists a deterministic algorithm that computes a symmetric PSD matrix $\tilde{L^+}$ such that $\tilde{L^+}\approx_{\epsilon}L^{+}$, and uses space $O(\log (N\cdot w)\log\log (N\cdot w/\eps))$, where $w=\wmax/\wmin$ is the ratio of the maximum and minimum edge weights in $G$ and $N$ is the bitlength of the input.
\end{namedtheorem}

In \cite{MRSV17}, the authors assume that multigraphs are given as a list of edges and hence the length of the input $N$ is always at least the sum of the edge weights. When moving to general weighted graphs, we will think of rational edge weights as being given as a numerator and a denominator each in binary. This is a more concise representation and hence incurs a space complexity that is dependent on the edge weights. 

The following claim about spectral approximation will be useful for proving \cref{lem:weighted_solver}

\begin{claim}
\label{claim:more_approx_facts}
Let $A,B,C$ be symmetric PSD matrices and let $0 < \eps,\eps' < 1$. Then:
\begin{enumerate}
\item If $A\approx_{\eps} B$ then $A^+\approx_{\eps/(1-\eps)} B^+$.
\item  If $A\approx_{\eps} B$ and $B\approx_{\eps'} C$ then $A\approx_{\eps+\eps'+\eps\cdot\eps'}C$.
\end{enumerate}
\end{claim} 
See \cite{MRSV17} for the proofs.\footnote{\cite{MRSV17} uses a variant of spectral approximation that measures closeness with $e^{\pm \eps}$ rather than $(1\pm\eps)$ but these two notions are equivalent up to constant factors and the proofs of these basic facts are nearly identical for both notions.} 

\begin{proof}[Proof sketch of \cref{lem:weighted_solver}]
The algorithm of \cite{MRSV17} has two steps, much like most of the efficient Laplacian solvers in the literature. Given the Laplacian $L$ of an undirected multigraph, they first compute a $c$-spectral approximation to $L^+$ for some constant $c<1/2$, using space $O(\log (N\cdot\wmax)\cdot\log\log (N\cdot\wmax))$ and then they boost this to an $\eps$-spectral approximation using an iterative method known as the Richardson iterations. This latter step uses an additional $O(\log (N\cdot\wmax)\cdot \log\log \frac{N\cdot\wmax}{\eps})$ space, giving the final space complexity. The first step is more delicate and is where the authors required integer edge weights, while Richardson iterations are agnostic to the edge weights and boosts any $c$-spectral approximation to the pseudoinverse to an $\eps$-spectral  approximation in space $O(\log (N\cdot\wmax)\cdot \log\log (N\cdot\wmax)/\eps)$, as long as $c<1/2$. 

Now, given a weighted Laplacian $L$, we will construct a multigraph having Laplacian $L'$ such that a scalar multiple of $(L')^+$ is a constant spectral approximation to $L^+$. Applying the first step of \cite{MRSV17} to $L'$ will allow us to compute a constant spectral approximation to $(L')^+$, which in turn will be a constant spectral approximation to $L^+$ by transitivity (Item (2) of \cref{claim:more_approx_facts}). With this, we can boost to an $\eps$-spectral approximation using Richardson iterations.

Set $z=\min\{1,w_{\mathrm{min}}\}$ where $w_{\mathrm{min}}$ is the minimum edge weight in $G$. Set $\delta=\frac{1}{6}$, $\gamma=\frac{\delta}{1-\delta}=\frac{1}{5}$ and $t=\left\lceil{\log \frac{2 n^3}{\delta  z}}\right\rceil$. We construct a multigraph $G'$ on $n$ vertices as follows. For each edge $(a,b)\in E(G)$, add an edge between $a$ and $b$ in $G'$ with weight $\lfloor{2^t\cdot w_{ab}}\rfloor$. Note that from the way we set $t$, these new weights are all positive integers. Let $L'=D'-A'$ be the Laplacian of $G'$ and note that for all $a,b\in[n]$ we have 
\begin{align*}
\left|2^{-t}\cdot A'_{ab}-A_{ab}\right|&=\left|2^{-t}\cdot \lfloor{2^t\cdot w_{ab}}\rfloor-w_{ab}\right|\\
&\leq 2^{-t}.
\end{align*}
Since for all $i$, $D_{ii} = \sum_{j}A_{ij}$ and $2^{-t}\cdot D'_{ii}=2^{-t}\cdot\sum_{j}A'_{ij}$, it follows that $|2^{-t}\cdot D'_{ii}-D_{ii}|\leq n\cdot 2^{-t}$. Letting $E=2^{-t}\cdot L'-L$, we can conclude that the sum of absolute values of the entries of each column of $E$ is bounded by $2\cdot n\cdot 2^{-t}$, denoted by $\| E\|_1 \le 2\cdot n\cdot 2^{-t}$. 

We will show that $2^{-t}\cdot L'\approx_{\delta} L$, which from \cref{claim:more_approx_facts} will imply that $2^t\cdot (L')^+=(2^{-t}\cdot L')^+\approx_{\gamma} L^+$. To see that $2^{-t}\cdot L'\approx_{\delta} L$, fix $v\in\mathbb{R}^n$. Since $L$ and $L'$ are both Laplacians of undirected graphs with the same connectivity status, they share a kernel. So for all vectors in the kernel, $L$ and $L'$ have equal quadratic forms and without loss of generality we can assume $v$ is in the orthogonal complement of the kernel. 

For any symmetric matrix $M$ we have that $|v^{\top}Mv|\leq v^{\top}v\cdot \|M\| \leq v^{\top}v\cdot \|M\|_1$. Setting $M=E$ and recalling our bound on the $\| \cdot \|_1$-norm of $E$ we get
\[
 \left|v^{\top}Ev\right|\leq v^{\top}v\cdot 2 n\cdot 2^{-t}.
\]
From \cref{claim:norm_of_pinv_stversion}, we have that $v^{\top}Lv\geq v^{\top}v\cdot \frac{z}{n^2}$.
Combining this with the above gives 

\begin{align*}
 \left|v^{\top}Ev\right|&\leq v^{\top}v\cdot 2 n\cdot 2^{-t}\\
 &\leq  v^{\top}v\cdot \delta \cdot \frac{z}{n^2}\\
 &\leq \delta\cdot v^{\top}Lv,
\end{align*}
from which it follows that $2^{-t}\cdot L'\approx_{\delta} L$. Thus, $2^t\cdot (L')^+\approx_{\gamma} L^+$. 

Using the first step in \cite{MRSV17}, we compute $\tilde{L}$, which is a $\gamma$-approximation to $2^t\cdot (L')^+$. Note that the maximum edge weight in $G'$ is upper-bounded by $\wmax\cdot 2^t=\poly(n\cdot \wmax/\wmin)$.  So the invocation of \cite{MRSV17} uses space $O(\log (N\cdot \wmax/\wmin)\cdot\log\log (N \cdot\wmax/\wmin))$.  From Item (2) of \cref{claim:more_approx_facts}, we have that $\tilde{L}\approx_{2 \gamma+\gamma^2} L^+$. Since $2 \gamma+\gamma^2<\frac{1}{2}$, we can apply Richardson iterations to $\tilde{L}$ to compute an $\eps$-spectral approximation to $L^+$ in total space $$O(\log (N\cdot\wmax/\wmin)\log\log (N\cdot\wmax/\eps\cdot\wmin)),$$ as desired.
\end{proof}

\end{document}